\documentclass{llncs}
\usepackage{amsmath}
\usepackage{amssymb}
\usepackage{latexsym}
\usepackage{epsfig}
\usepackage{euscript}
\usepackage{tikz}
\begin{document}
\title{Word-Mappings of level $3$}
\author{G. S\'enizergues}
\institute
{LaBRI and 
Universit\'e de Bordeaux 
\footnote{
mailing adress:LaBRI and UFR info, Universit\'e de Bordeaux\\
               351 Cours de la lib\'eration -33405- Talence Cedex.\\
email:\{geraud.senizergues\}@u-bordeaux.fr; 
fax: 05-56-84-66-69;\\
URL:http://dept-info.labri.u-bordeaux.fr/$\sim$ges
}
}

\newtheorem{rem}[theorem]{Remark}
\newenvironment{sketch}[0]{\noindent{\bf Sketch of proof}:}{$\Box$\\}
\newenvironment{constru}[0]{\noindent{\bf Construction}:}{$\Box$\\}
\newenvironment{exmp}[0]{\noindent{\bf Example}:}{$\Box$\\}

\newcommand{\bbbB}{{\mathbb B}}
\newcommand{\bbbS}{{\mathbb S}}
\newcommand{\bbbD}{{\sf D}}
\newcommand{\bbbK}{{\sf K}}
\newcommand{\bbbM}{{\sf M}}
\newcommand{\bbbN}{{\mathbb N}}
\newcommand{\bbbR}{{\sf R}}
\newcommand{\bbbQ}{{\mathbb Q}}
\newcommand{\bbbV}{{\sf V}}
\newcommand{\bbbZ}{{\mathbb Z}}

\newcommand{\Bin}{{\bf \mathcal B}}
\newcommand{\card}{{\rm Card}}
\newcommand{\mi}{{\rm min}}
\newcommand{\borninf}{{\rm inf}}
\newcommand{\ma}{{\rm max}}
\newcommand{\rd}{{\rm rd}}
\newcommand{\di}{{\rm dim}}
\newcommand{\Div}{{\rm Div}}
\newcommand{\dom}{{\rm dom}}
\newcommand{\im}{{\rm im}}

\newcommand{\AutA}{{\mathbb A}}
\newcommand{\BB}{\mathbb{B}}
\newcommand{\CC}{\mathbb{C}}
\newcommand{\DD}{\mathbb{D}}
\newcommand{\A}{\mathcal{A}}
\newcommand{\B}{\mathcal{B}}
\newcommand{\C}{\mathcal{C}}
\newcommand{\GG}{\mathcal{G}}
\newcommand{\Id}{{\rm Id}}
\newcommand{\MM}{{\mathbb M}}
\newcommand{\NN}{{\mathbb N}}
\newcommand{\OMEG}{{\mathbb O}}
\newcommand{\Pd}{{\mathbb P}{\mathbb d}}
\newcommand{\Pow}{{\mathcal P}} 
\newcommand{\U}{{\mathcal U}}
\newcommand{\T}{\mathcal{T}}
\newcommand{\Un}{{\bf 1}}
\newcommand{\apda}{\operatorname{apda}}
\newcommand{\pda}{\operatorname{pda}}
\newcommand{\dpda}{\operatorname{dpda}}
\newcommand{\sdpda}{\operatorname{sdpda}}

\newcommand{\BR}{\operatorname{BR}}
\newcommand{\CF}{\operatorname{CF}}
\newcommand{\CHG}{\operatorname{CHG}}
\newcommand{\CP}{\operatorname{CP}}
\newcommand{\FL}{\operatorname{FL}}
\newcommand{\HOM}{\operatorname{HOM}}
\newcommand{\IRR}{\operatorname{IRR}}
\newcommand{\MAC}{\operatorname{MAC}}
\newcommand{\OP}{\operatorname{OP}}
\newcommand{\PB}{\operatorname{PB}}
\newcommand{\FS}{{\cal P}}
\newcommand{\POP}{\operatorname{POP}}
\newcommand{\PUSH}{\operatorname{PUSH}}
\newcommand{\RAT}{\operatorname{RAT}}
\newcommand{\REG}{\operatorname{REG}}
\newcommand{\SOL}{\operatorname{SOL}}
\newcommand{\TR}{\operatorname{TR}}
\newcommand{\TOPSYMS}{\operatorname{TOPSYMS}}

\newcommand{\Al}{\operatorname{Alph}}
\newcommand{\al}{\operatorname{al}}
\newcommand{\cp}{\operatorname{copy}}
\newcommand{\lp}{\langle~\langle}
\newcommand{\rp}{\rangle~\rangle}
\newcommand{\pa}[1]{\langle\langle~#1~\rangle\rangle}
\newcommand{\papol}[1]{\langle~#1~\rangle}
\newcommand{\dd}{{\rm d}}
\newcommand{\CO}{{\rm Cong}}
\newcommand{\cpds}{\operatorname{cpds}}
\newcommand{\DIFF}{{\mathcal D}}
\newcommand{\elsop}{\operatorname{else}}
\newcommand{\flag}{f}
\newcommand{\fop}{{\rm f}}
\newcommand{\FRAC}{{\mathcal F}}
\newcommand{\HYP}{{\mathcal H}}
\newcommand{\ifop}{\operatorname{if}}
\newcommand{\iter}{\operatorname{it}}
\newcommand{\Language}{{\rm L}}
\newcommand{\LYP}{{\mathcal L}}
\newcommand{\N}{{\rm N}}
\newcommand{\norm}[1]{\parallel#1\parallel}
\newcommand{\op}{\operatorname{op}}
\newcommand{\pds}{\operatorname{pds}}
\newcommand{\ptcup}{\stackrel{.}{\cup}}
\newcommand{\chg}{\operatorname{chg}}
\newcommand{\pop}{\operatorname{pop}}
\newcommand{\push}{\operatorname{push}}
\newcommand{\rest}{r}
\newcommand{\sep}{; }
\newcommand{\Str}{{\mathcal S}}
\newcommand{\supp}{{\rm supp}}
\newcommand{\Sys}{{\mathcal S}}
\newcommand{\term}{\operatorname{term}}
\newcommand{\uterm}{\operatorname{uterm}}
\newcommand{\Terms}[1]{{#1}-\term(\Gamma \cup {\mathcal U})}
\newcommand{\UTerms}[1]{{#1}-\uterm(\Gamma \cup {\mathcal U})}
\newcommand{\thenop}{\operatorname{then}}
\newcommand{\topo}{{\mathcal T}}
\newcommand{\tops}{\operatorname{tops}}
\newcommand{\topsyms}{\operatorname{topsyms}}
\newcommand{\transpose}[1]{\;\stackrel{t}{\,}\!\!#1}
\newcommand{\var}{\operatorname{var}} 
\newcommand{\vterm}{\operatorname{vterm}} 
\newcommand{\WD}{\operatorname{WD}}

\newcommand{\Alg}{{\cal A}}
\newcommand{\Mod}{{\cal M}}
\newcommand{\ff}{\bot}
\newcommand{\cd}{\cdots}
\newcommand{\Clos}{{\rm Cl}}
\newcommand{\imp}{\rightarrow}
\newcommand{\eq}{\leftrightarrow}
\newcommand{\ou}{\vee}
\newcommand{\et}{\wedge}
\newcommand{\non}{\neg}
\newcommand{\I}[1]{{\rm I}(#1)}
\newcommand{\DB}{\operatorname{DB}}
\newcommand{\vir}{\underline{,}}
\newcommand{\Form}{{\cal L}_1}
\newcommand{\Term}{{\cal T}}
\newcommand{\GForm}{{\cal G}_1}
\newcommand{\FVar}{{\rm FV}}
\newcommand{\Suc}{{\rm Succ}}
\newcommand{\spc}{\operatorname{spc}}
\newcommand{\sps}{\operatorname{sps}
}
\newcommand{\E}{{\sf E}}
\newcommand{\DIF}{{\sf \Delta}}
\newcommand{\SUM}{{\sf \Sigma}}
\newcommand{\COMPseq}{{\circ}}
\newcommand{\COMPser}{{\bullet}}
\def\nbOne{{\mathchoice {\rm 1\mskip-4mu l} {\rm 1\mskip-4mu l} {\rm
1\mskip-4.5mu l} {\rm 1\mskip-5mu l}}}

\newcommand{\eqbar}{\:|\!\!\!=\!\!\!\!=\!\!\!|\:}
\newcommand{\deducedir}{\mbox{$\:|\!\!\!-\!\!\!-\:$}}
\newcommand{\deducedirnobox}{\:|\!\!\!-\!\!-\:}
\newcommand{\ededucedir}{\mbox{$\:|\!\:|\!\!\!-\!\!\!-\:$}}
\newcommand{\ededucedirnobox}{\:|\!\deducedirnobox}
\newcommand{\antiededucedir}{-\!\!-\!\!\!||}
\newcommand{\deduce}[1]{\stackrel{<#1>}{\deducedirnobox}}
\newcommand{\drule}[2]{\frac{\underline{#1}}{#2}}
\newcommand{\dual}[1]{{\rm D}(#1)}
\newcommand{\ededuce}[1]{\stackrel{<#1>}{\ededucedirnobox}}
\newcommand{\antiededuce}[1]{\stackrel{<#1>}{\antiededucedir}}
\newcommand{\deducesq}[1]{\stackrel{[#1]}{\deducedirnobox}}
\newcommand{\ededucesq}[1]{\stackrel{[#1]}{\ededucedirnobox}}
\newcommand{\eqsem}{\:|\!\!\!=\!\!\!\!=\!\!\!|\:}
\newcommand{\force}{\ededucedir}
\newcommand{\notforce}{\not{\!\!\!\!\force}}
\newcommand{\deducesem}{\:|\!\!\!=\!\!\!\!=\:}
\newcommand{\equivsem}{\:|\!\!\!=\!\!\!\!=\!\!\!|\:}
\newcommand{\notdeducesem}{\not{\!\!\!\!\deducesem}}
\newcommand{\these}{\deducedir}

\newcommand{\deralt}[1]{{\rightarrow}^*_{#1}\;}

\newcommand{\path}[2]{\stackrel{#1}{\rightarrow_{#2}}\;}

\newcommand{\derivedir}[1]{{\rightarrow}_{#1}\;}

\newcommand{\derive}[2]{{\rightarrow}^{#1}_{#2}\;}

\newcommand{\derivestar}[1]{{\rightarrow}^*_{#1}\;}

\newcommand{\computdir}[1]{{\vdash}_{#1}\;}

\newcommand{\comput}[2]{{\vdash}_{#1}^{#2}\;}

\newcommand{\computstar}[1]{{\vdash}^*_{#1}\;}

\newcommand{\qcomputdir}[1]{{\twoheadrightarrow}_{#1}\;}

\newcommand{\qcomput}[2]{{\twoheadrightarrow}^{#1}_{#2}\;}

\newcommand{\qcomputstar}[1]{{\twoheadrightarrow}^*_{#1}\;}

\newcommand{\CFT}{\operatorname{CFT}}
\newcommand{\D}{\operatorname{D}}
\newcommand{\DPT}{\operatorname{DPT}}
\newcommand{\Pil}{\operatorname{P}}
\newcommand{\PebT}{\operatorname{PT}}
\newcommand{\PilT}{\operatorname{PT}}
\newcommand{\RT}{\operatorname{RT}}
\maketitle
\begin{abstract}
Sequences of numbers (either natural integers, or integers or rational)
of level $k \in \bbbn$ have been defined in \cite{Fra05,Fra-Sen06} as the
sequences which can be computed by deterministic pushdown automata of level $k$.
This definition has been extended to sequences of {\em words} indexed by {\em words}
in \cite{Sen07,Fer-Mar-Sen14}.
We characterise here the sequences of level 3 as the compositions of two  
HDT0L-systems. Two applications are derived:\\
- the sequences of rational numbers of level 3 are characterised by polynomial
recurrences \\
- the equality problem for sequences of rational numbers of level 3 is decidable.
\end{abstract}

\noindent{\bf Keywords:}$\;\;$Iterated pushdown automata \sep 
recurrent sequences of numbers\sep equivalence problems.

\pagenumbering{arabic}
\noindent {\small{\bf Version}: \today }

\section{Introduction}
\label{introduction}
The class of pushdown automata of level $k$ (for $k \geq 1$) has been
introduced in \cite{Gre70},\cite{Maslov74} as a generalisation of the
automata and grammars of \cite{Aho68},\cite{Aho69},
\cite{Fischer68} and has been the object of many further studies: 
see \cite{Maslov76}, \cite{Eng-Sch77}, \cite{Damm82}, \cite{Engel83}, \cite{Eng-Slut84},
\cite{Eng-Vog86}, \cite{Damm&Goerdt86},
and more recently \cite{Cau02}, \cite{Knapik},
\cite{Car-Wor03},\cite{Fra05},\cite{HagueOng07},\cite{Car-Ser21}.\\

The class of {\em integer} sequences computed
(in a suitable sense) by such automata was defined in \cite{Fra05},\cite{Fra-Sen06} (we denote it by $\mathbb{S}_k$).
The class of {\em word} mappings from a free monoid $A^*$ into a free monoid $B^*$  computed
by such automata was defined in \cite{Sen07} (as a straightforward extension of $\mathbb{S}_k$); we denote it by $\mathbb{S}_k(A^*, B^*)$.

Let $\FRAC(S_k(A^*,\NN))$ be the set of all the sequences of 
{\em rational} numbers which can be decomposed as $\frac{a_n - b_n}{a'_n - b'_n}$ for sequences $a,b,a',b' \in \mathbb{S}_k(A^*,\NN)$
(Definition 9 of \cite{Sen07}).\\
These classes of number sequences fulfill many closure properties and generalize some well-known
classes of recurrent sequences or formal power series(see \cite{Fra-Sen06}, sections 4,6,7).\\

This paper focuses on level 3, for words i.e. $\mathbb{S}_3(A^*, B^*$).
The class $\FRAC(S_3(\NN,\NN))$ contains all the so-called P-recurrent sequences
of rational numbers, corresponding also to the D-finite formal power series
(see \cite{Sta80} for a survey and \cite{Pet-Wil-Zei96} for a thorough study of their algorithmic properties).

The main result of this paper consists in characterising the class $\mathbb{S}_3(A^*, B^*)$  in terms of Lindenmayer systems i.e.
iterated homomorphisms.\\
{\bf Theorem \ref{characterisation_level3}}:\\
The following properties are equivalent:\\
1- $f \in \mathbb{S}_3(A^*,B^*)$\\
2- There exists a finite family $(H_i)_{i\in [1,n]}$ of mappings $A^* \rightarrow Hom(C^*,C^*)$ 
which fulfils a system of recurrent relations in $(Hom(C^*,C^*),\circ,{\rm Id})$ ,
an element $h \in Hom(C^*,B^*)$ and a letter $c \in C$ such that, for every $w \in A^*$:
$$f(w)=h(H_1(w)(c)).$$
3- $f$ is a composition of a DT0L sequence $g: A^* \rightarrow C^*$ by a HDT0L sequence $h: C^* \rightarrow B^*$.

As a corollary, the class $\mathbb{S}_3(A^*,\NN)$ is characterised by {\em polynomial recurrences}. 
The equality problem for two sequences in $\FRAC(S_3(A^*,\NN))$ can thus be solved by a 
suitable reduction to polynomial ideal theory (Theorem 5 of \cite{Sen07}).\\

This theorem \ref{characterisation_level3}, as well as a general version treating the class $\mathbb{S}_k(A^*,B^*)$, for every $k \geq 2$,
was announced in Theorem 6 of the extended abstract \cite{Sen07}. The full proof for $k=3$ is given here for the first time.\\
The main difficulty is to prove that $(1) \Rightarrow (2)$. The proof consists in constructing from the given pushdown-automaton of level 3,
a sequence of homomorphisms that fulfills a monoidal recurrence, in the sense of \cite{Fer-Mar-Sen14}. The main theorem of this reference
(slightly reformulated  here as Theorem \ref{characterisation_level2} in section \ref{sec-prelimin2}) then gives the conclusion.

\tableofcontents
\section{Preliminaries}
\label{sec-prelimin2}
We recall/introduce in \S \ref{subsec-sets}-\ref{subsec-rec-monoids}  some notation and basic definitions which will be used
throughout the text. In \S \ref{subsec-level2} we recall previous results on sequences of level 2, which are crucial for this work.
For overviews on the links between automata-theory and number-theory, we forward the reader to \cite{All-Sha03,Rig06}.
\subsection{Sets-Relations}
\label{subsec-sets}
Given a set $E$, we denote by $\Pow(E)$ the set of its subsets and by ${\mathcal P}_f(E)$ the set of its {\em finite } subsets.\\
A {\em binary relation} from a set $E$ into a set $F$ is a subset $R$ of $E \times F$.
The domain and image of $R$ are defined by:
$$\dom(R) := \{x \in E \mid \exists y \in F, (x,y) \in R\},\;\;
\im(R) := \{y \in F \mid \exists x \in E, (x,y) \in R\}.$$
We denote by $\circ$ the composition of binary relations:
if $R \subseteq E \times F, R'\subseteq F \times G$ then:
$$R \circ R':= \{ (x,z) \in E \times G \mid \exists y \in F, (x,y) \in R \;\et\; (y,z) \in R'\}$$
We denote by $\BR(Q)$ the set of binary relations over $Q$.\\
A {\em function} from the set $E$ into the set $F$ is a binary relation 
$f \subseteq E \times F$ such that,
$$\forall (x,y), \forall (x',y') \in f, x=x' \Rightarrow y=y'$$
Note that, when using a functional notation, we still use the composition operator $\circ$ as above i.e.
$$(f \circ g)(x) := g(f(x)).$$
We call {\em mapping} (or simply, {\em map}) from $E$ to $F$ any function $f: E \rightarrow F$ such that $\dom(f) = E$.\\
The empty set $\emptyset$ is also denoted by the symbol $\varepsilon$ when it is viewed as the empty word.
\subsection{Abstract rewriting}
\label{preli_abstract_rewriting}
Let $E$ be some set and $\rightarrow \subseteq E \times E$. 
The relations $\rightarrow^{n}$ (for any natural integer $n$) and $\rightarrow^{*}$ are
defined from the binary relation $\rightarrow$ as usual (see \cite{Hue80}).
A {\em derivation} with respect to the relation $\to$ is a sequence :
  $$D= (e_i)_{i \in I}$$
of elements $e_i \in E$, in dexed by $I=[0,n]$ or $I = \NN$,  such that, for every $i\in \NN$, if $i+1 \in I$ then
$e_i \to e_{i+1}$.
The {\em length} of derivation $D$ is the integer $n$ (if $I=[0,n]$)  or $\infty$ (if $I=\NN$).
The notation  $e \to^{\infty}$ means that there exists some  derivation of length $\infty$ that starts on $e$.
The relation $\to$ is called {\em noetherian} iff there exists no $e \in E$ such that 
$e \to^{\infty}$. 

\subsection{Monoids}
\label{preli_monoids}
We recall a monoid is a triple $\langle M,\cdot,\Un \rangle$ such that,
$M$ is a set (the carrier of the monoid), $\cdot$ is a composition law which is associative and $\Un$ is a neutral element for $\cdot$ (on both sides).
Given two monoids
$\MM_1:=\langle M_1,\cdot,\Un_1\rangle,\;\;\MM_2:=\langle M_2,\cdot,\Un_2\rangle $ a 
monoid-homomorphism from $\MM_1$ to $\MM_2$ is a map $h: M_1 \rightarrow M_2$
fulfilling: for every $x,y \in  M_1$
$$h(x \cdot y) = h(x) \cdot h(y) \mbox{ and } h(\Un_1) = \Un_2.$$
We denote by $\HOM(\MM_1,\MM_2)$ the set of all monoid-homomorphism from $\MM_1$ to $\MM_2$.
For every monoid $\MM$, the set $\HOM(\MM,\MM)$, endowed with the composition law $\circ$ and the identity map $\Id_M$ is a monoid.\\
Given a set $X$ (that we see as an ``alphabet''), we denote by $X^*$ the set of all 
finite words labelled on this set $X$. We denote by $\cdot$ the binary operation
of concatenation over $X^*$ 
and denote by $\varepsilon$ the empty word.
The structure $\langle X^*, \cdot,\varepsilon \rangle$ is the {\em free monoid} 
over the alphabet $X$.
For every integer $k \geq 0$, by $X^{\leq k}$ we mean the set $\{ u \in X^* \mid |u| \leq k\}$.\\
\paragraph{Partial monoids}
We call {\em partial-monoid} every triple $\MM = \langle M,\cdot,\Un \rangle$ such that,
$M$ is a set (the carrier of the partial-monoid), $\cdot$ is a function from $M \times M$ into $M$ such that:
for every $x,y,z \in M$
$$ [(x,y) \in \dom(\cdot) \mbox{ and } (x\cdot y,z ) \in \dom (\cdot)] \Leftrightarrow
[(y,z) \in \dom(\cdot) \mbox{ and } (x, y\cdot z ) \in \dom (\cdot)]$$
and, if $x,y,z$ fulfill the above two equivalent prerequisite, then
$$ (x \cdot y) \cdot z = x\cdot (y \cdot z)$$
and, for every $x \in M$
$$x\cdot \Un = \Un \cdot x = x.$$
Note that, if $\langle M,\cdot,\Un \rangle$ is a partial monoid, then
$$\langle \Pow(M),\cdot,\{\Un\} \rangle$$ is a monoid
(here , for $X,Y \subseteq M, X \cdot Y =\{ x\cdot y \mid x \in X,y \in Y, (x,y) \in \dom(\cdot)\}$). Let
us consider the monoid $\MM^0$ defined by
$$\MM^0:= \langle \Pow_1(M),\cdot,\{\Un\} \rangle$$ where $\Pow_1(M)$ denotes the set of subsets of $M$ with at most one element.
$\MM^0$  is a sub-monoid of the monoid $\langle \Pow(M),\cdot,\{\Un\} \rangle$. When $\MM$ is a partial-monoid which fails to be a monoid, there is a bijection between $\Pow_1(M)$ and $M \cup \{0\}$,
where $0$ is a new element not in $M$. Most assertions about a partial monoid $\MM$  can thus be deduced from the similar assertions about the monoid 
$\MM^0$.
\paragraph{Congruences}
\label{par-prelimin-congruences}
An equivalence relation $\sim$ over $M$ is called a {\em right-regular} equivalence
if and only if, for every $x, y , z \in M$
$$x \sim y \Rightarrow x\cdot z \sim y \cdot z.$$
(the notion of left-regular equivalence is defined analogously).\\
An equivalence over the monoid $M$ is called a {\em congruence} iff it is both right-regular and left-regular.\\
Given an homomorphism $h: \MM_1 \rightarrow \MM_2$, the kernel of $h$ is the congruence $\equiv$ over $\MM_1$ defined by:
$$x \equiv y \;\;\Leftrightarrow \;\;h(x)=h(y).$$
We shall use the following construction.
Let $\MM_1,\MM_2$ be monoids and let us consider their free product $\MM_1 * \MM_2$.
Let $h:\MM_1 \rightarrow \MM_2$ be an homomorphism. Let $\hat{h}: \MM_1 * \MM_2 \rightarrow M_2$ be the unique monoid-homomorphism such that:
$$\forall x_1 \in M_1,\;\;\hat{h}(x_1) = h(x_1),\;\;\forall  x_2\in M_2,\;\;\hat{h}(x_2)= x_2.$$
The definition of $\hat{h}$ shows that:
$$\forall x,y \in M_1, (x,y) \in \ker(\hat{h}) \Leftrightarrow (x,y) \in \ker(h). $$
Therefore $\ker(\hat{h})$ is an extension of the congruence $\ker(h)$ to the monoid $\MM_1 * \MM_2$.\\
Starting with a congruence $\equiv$ over $\MM_1$, we can define $\MM_2 := \MM_1/\equiv$, consider the projection
$h: \MM_1 \rightarrow \MM_2$ and define an extension $\hat{\equiv}$ of $\equiv$  over the monoid $\MM_1 * (\MM_1/\equiv)$
by
\begin{equation}
\hat{\equiv} := \ker(\hat{h}).
\label{eq-extension-congruence}
\end{equation}
We call $\hat{\equiv}$ the self-extension of $\equiv$.
The monoid $\MM_1/\equiv$ is isomorphic with $(\MM_1 * (\MM_1/\equiv))/\hat{\equiv}$, by the map
$$[x]_\equiv \mapsto [x]_{\hat{\equiv}}.$$
Intuitively, the elements of $\MM_1 * (\MM_1/\equiv)$ are ``mixed products'' of ``genuine'' elements of $M_1$ with ``fake'' elements
that are merely congruence-classes. When we map the genuine elements onto their classes, what we obtain is 
a product of classes. The value of this product  is the class of the
``mixed product'' for the self-extension $\hat{\equiv}$.\\
This construction will be useful for extending some natural congruences over ordinary pushdowns of order 3 to
pushdowns of order 3 {\em extended with undeterminates} (see subsection \ref{subsub-termination}), where the undeterminates
are classes of ordinary pushdowns modulo some congruence.

\subsection{General Automata }
\label{sub_general_nondetautomata}
\paragraph{General automaton}
We introduce in this paragraph a notion of automaton over
an arbitrary memory structure, which is essentially the one defined by \cite{Eng91}. 
Let us call {\em data-structure} every tuple 
$\DD = \langle D, F, \OP, \tops \rangle$ consisting of a set 
$D$, a finite set $F$, a set $\OP$ of unary operations (
$\forall \op \in \OP$, $\op$ is a total map : $D \rightarrow D$) 
 and a total map $\tops: D \rightarrow F$.\\
\begin{definition}[Automaton over $\DD$]
An automaton over the data-structure $\DD$ and 
the terminal alphabet $B$ is a 4-tuple
    $$\mathcal{A}=(Q,B,\DD,\delta)$$ where
$Q$ is a finite set of states,\\
$\DD= \langle D, F, \OP, \tops \rangle$ is a data-structure,\\
the transition function, $\delta$, is a map from $Q \times (B\cup\{\varepsilon\}) \times F$ 
into $\FS(Q \times \OP(\DD))$.
\label{def-syntax_pda}
\end{definition}
A configuration is a triple $(q,u,w) \in Q \times B^* \times D$.
The {\em direct computation relation}  $\computdir{\A}$, is defined as usual:
it is a binary relation over the set of configurations.For every $q,q'\in Q, u,u'\in B^*, d,d'\in D$,
we let
$$(q,u,d)\computdir{\A}(q',u',d')$$ iff
\begin{itemize}
\item $u=\bar{b} \cdot u'$ for some $\bar{b} \in B \cup \{\varepsilon\}$
\item $\exists \op \in \OP , (q', \op) \in \delta(q,\bar{b},\tops(d))$ and $d'= \op(d)$.
\end{itemize}

\paragraph{Pushdowns}
Given a data-structure $\DD$ and a non-empty  alphabet $\Gamma$, one constructs a new data-structure, denoted by $\Pd(\Gamma,\DD)$ and called the
set of {\em pushdowns over} $(\Gamma,\DD)$. It is defined as follows:\\
$$ \Pd(\Gamma,\DD) := \langle P(\Gamma,D), F', \OP', \tops' \rangle$$
where
\begin{eqnarray*}
P(\Gamma,D) & := & (\Gamma[D])^*,\;\; F':= \Gamma \cdot F \cup \Gamma \cup \{\varepsilon\},\;\;
\OP':= \OP \cup \{ \pop\}\cup \{ \push(h) \mid h \in \Gamma^+\},\\
&&\tops'(\gamma[d]w) := \gamma \cdot \tops(d),\;\;\tops(\varepsilon):= \varepsilon.
\end{eqnarray*}
The maps $\push(h): P(\Gamma,D) \rightarrow P(\Gamma,D)$ and $\pop: P(\Gamma,D)  \rightarrow P(\Gamma,D)$
are defined by: for every $h \in \Gamma^+, d \in D, w \in D^*$:if $h=h_1\cdots h_n$ with $h_i \in \Gamma$
$$ \push(h)(\gamma[d]w) := h_1[d]\cdots h_n[d]w,\;\;\pop(\gamma[d]w):=w,$$
$$\;\;\push(h)(\varepsilon) := \varepsilon,\;\;\pop(\varepsilon):=\varepsilon.$$
A {\em pushdown automaton}
of level $k$, for some integer $k \geq 0$, is an automaton over the data-structure $\DD_k$ defined by:
$$\DD_0 := \langle \{\emptyset\},\{\emptyset\},\emptyset,\{(\emptyset,\emptyset)\}\rangle\;\;\mbox{ and }\;\;\forall i \in \NN,\;\;\DD_{i+1} := \Pd(\Gamma,\DD_i).$$
The usual notion of finite automaton (resp. pushdown-automaton) coincides with the above notion of pda of level $0$ (resp. $1$). 
\paragraph{Derivations versus computations}
\label{par-derive-vs-comput}
Let $\DD$ be some data-structure and $$\mathcal{A}=(Q,B,\Pd(\Gamma,\DD),\delta)$$ a pushdown-automaton i.e. an automaton over the data-structure $\Pd(\Gamma,\DD)$.
We associate with ${\mathcal A}$ the infinite ``alphabet''
\begin{equation} 
V_{\mathcal A}=\{( p, \omega, q) \mid p,q \in Q, \omega \in
P(\Gamma,D)-\{\varepsilon\}\}.
\label{alphabet_VA}
\end{equation}
whose elements are called {\em variables}.
The set of {\em productions } associated with ${\mathcal A}$, denoted by
$P_{\mathcal A}$ is made of the set of all the following rules:\\
the {\em transition} rules: 
$$ ( p, \omega, q) \derivedir{\mathcal A} \bar{a} ( p', \omega', q)$$
if $(p ,\bar{a},\omega) \computdir{\mathcal A} ( p',\varepsilon,\omega')$
and $q \in Q $ is arbitrary,
$$ ( p, \omega, q) \derivedir{\mathcal A} \bar{a} $$
if $(p ,\bar{a},\omega) \computdir{\mathcal A} ( q,\varepsilon,\varepsilon)$.\\
the {\em decomposition} rule:
$$ ( p ,\omega, q) \derivedir{\mathcal A} ( p, \eta, r) ( r, \eta', q)$$ 
if $\omega= \eta \cdot \eta', \eta \neq \varepsilon, \eta' \neq \varepsilon$ and $p,q,r\in Q$ are arbitrary.
The one-step {\em derivation} generated by ${\mathcal A}$, denoted by 
$\derivedir{{\mathcal A}}$, is the smallest subset of $(V \cup \Sigma)^* \times (V \cup \Sigma)^*$ which contains $P_{\mathcal A}$ and is compatible with left-product and  right-product.
Finally, the {\em derivation} generated by ${\mathcal A}$, denoted by 
$\derivestar{{\mathcal A}}$, is the reflexive and transitive closure of
$\derivedir{{\mathcal A}}$.
These notions correspond to the usual notion of {\em context-free grammar}
associated with the pushdown automaton ${\mathcal A}$.
As soon as $D$ is infinite, this variable alphabet is infinite,
but all the usual properties of the relation $\derivedir{{\mathcal A}}$ and its links with
$\computdir{{\mathcal A}}$ remain true in
this context (see \cite[proof of theorem 5.4.3, pages 151-158]{Har78}).
In particular, for every $u \in \Sigma^*, p,q \in Q, \omega \in P(D)$
$$ (p,\omega,q) \derivestar{{\mathcal A}} u \Leftrightarrow (p,u,\omega) \computstar{{\mathcal A}}
(q,\varepsilon,\varepsilon).$$
We usually assume that $P(D)$ and $Q$ are disjoint, therefore, omitting the commas in $(p,\omega,q)$ does not lead to any confusion.
The notions of {\em derivation} $\pmod{\derivedir{{\mathcal A}}}$ is obtained by instantiating the general notion defined
in subsection \ref{preli_abstract_rewriting}.

\subsection{Pushdown automata of level $k$}
\label{sub_automata}
\paragraph{pda}
Beside the usual notions of finite automaton and pushdown automaton, we shall
consider here the notion of {\em pushdown automaton of level $k$}:
this is an automaton (in the sense of Definition \ref{def-syntax_pda})
over the data-structure $\DD_k$ (as defined in subsection \ref{sub_general_nondetautomata}).
Let us describe in more details these automata:
\begin{definition}[$k$-iterated pushdown store]
Let $\Gamma$ be a set. We define inductively the set of \emph{ $k$-iterated pushdown-stores over
$\Gamma$} by:
$$
0-\pds(\Gamma)=\{ \varepsilon \} \; \; 
(k+1)-\pds(\Gamma)=(\Gamma[k-\pds(\Gamma)])^*\;\;
\iter-\pds(\Gamma)= \bigcup_{k \geq 0} k-\pds(\Gamma).
$$
\label{def-kpds}
\end{definition}
The elementary operations that a $k$-pda can perform are:\\
- $\pop_j$ of level $j$ (where $1 \leq j \leq k$), which consists of 
popping the leftmost letter of level $j$ and all the bracket just on its right\\
- $\push_j(h)$ of level $j$ (where $1 \leq j \leq k,h \in \Gamma^+$), which consists of 
pushing successively all the symbols of $h$ as new heads of the leftmost pushdown of level $j$, thus copying this pushdown
in each bracket following each new head on its right(see example \ref{ex-pushdowns1} below).\\
An  operation of level $j$ (whether $\pop_j$ or $\push_j(h)$) applied on a pushdown where this leftmost pushdown of level $j$ is $\varepsilon$,
leaves the pushdown invariant. 

\begin{example}
\label{ex-pushdowns1}
Let $\Gamma=\{ A_1,A_2,A_3,B_1,B_2,B_3,C_1,C_2,C_3,D_1,D_2,D_3\}.$
Let $$\omega=A_1[A_2[A_3[\varepsilon]C_3[\varepsilon]]B_2[D_3[\varepsilon]C_3[\varepsilon]]]B_1[B_2[B_3[\varepsilon]D_3[\varepsilon]]].$$
We shall (abusively) write:
$$\omega=A_1[A_2[A_3C_3]B_2[D_3C_3]]B_1[B_2[B_3D_3]].$$
i.e. we use a short notation where the bracketed empty words  $[\varepsilon]$ that occur at the most internal level of the pushdown (here level 4) are removed.\\
The reading operation, applied on the above example gives:\\
$\topsyms(\omega)=A_1A_2A_3$\\
The $\pop$ operations, applied on the above example give:
\\
$\pop_1(\omega)=B_1[B_2[B_3D_3]]$\\
$\pop_2(\omega)=A_1[B_2[D_3C_3]]B_1[B_2[B_3D_3]]$\\
$\pop_3(\omega)=A_1[A_2[C_3]B_2[D_3C_3]]B_1[B_2[B_3D_3]]$\\
The $\push$ operations, applied on the above example gives:
\\
$\push_1(AB)(\omega)=A[A_2[A_3C_3]B_2[D_3C_3]]B[A_2[A_3C_3]B_2[D_3C_3]]B_1[B_2[B_3D_3]]$\\
$\push_2(AB)(\omega)=A_1[A[A_3C_3]B[A_3C_3]B_2[D_3C_3]]B_1[B_2[B_3D_3]]$\\
$\push_3(AB)(\omega)=A_1[A_2[ABC_3]B_2[D_3C_3]]B_1[B_2[B_3D_3]]$
\end{example}

A transition of the automaton consists, given the word $\gamma$ made of all the leftmost letters 
of the $k$-pushdown (the one of level 1, followed by the one of level 2, ..., followed by the one of level $k$),
the state $q$ and the leftmost letter $b$ (or, possibly, the empty word $\varepsilon$) on the input tape,
in performing one of the above elementary operations. More formally, 
\begin{definition}[operations on $k$-$pds$]
Let $k \geq 1$ , let $\POP(k):=\{ \pop_j | j \in [1,k]   \}$,
$\PUSH(k,\Gamma):=\{ \push_j(\gamma) | \gamma \in \Gamma^+ , j \in [1,k]\}$, $\OP(k,\Gamma) = \POP(k)  \cup \PUSH(k,\Gamma)$ ,
$\TOPSYMS(k,\Gamma):=\Gamma^{\leq k} -\{\varepsilon\}$.
\label{def-kops}
\end{definition}

\begin{definition}[$k$-$pdas$]
A k-iterated pushdown automaton (abbreviated $k$-pda) over a terminal alphabet $B$ is a 4-tuple
    $$\mathcal{A}=(Q,B,\Gamma,\delta)$$ where
    \begin{itemize}
        \item $Q$ is a finite set of states,
        \item $\Gamma$ is a finite set of pushdown-symbols
        \item the transition function $\delta$ is a map from $Q \times
        (B \cup \{\varepsilon \}) \times \TOPSYMS(k,\Gamma)$
        into the set of finite subsets of $Q \times \OP(k,\Gamma)$ 
      \end{itemize}
\label{def-syntax_kpda}
\end{definition}
The automaton ${\mathcal A}$ is said {\em deterministic} iff, 
for every $q \in Q, \gamma \in \Gamma^{\leq k},b \in B$\\
\begin{equation}
\card (\delta (q, \varepsilon, \gamma)) \leq 1 \mbox{ and }\card(\delta (q,b,\gamma)) \leq 1,
\label{determinism0}
\end{equation}
\begin{equation}
\card (\delta (q, \varepsilon,\gamma)) = 1 \Rightarrow
\card (\delta(q, b,\gamma)) = 0.
\label{determinism1}
\end{equation}
\paragraph{Normalized automata}
\label{par-normal-aut}
In order to define a useful notion of map {\em computed} by a $k$-pda we introduce the following stronger condition: ${\mathcal A}$ is called {\em strongly deterministic} iff, 
for every $q \in Q, \gamma \in \Gamma^{(k)}-\{\varepsilon\}$\\
\begin{equation}
\sum_{\bar{b} \in \{\varepsilon\} \cup B} \card (\delta (q, \bar{b}, \gamma)) \leq 1 
\label{strongdeterminism}
\end{equation}
In other words, the automaton ${\mathcal A}$ is {\em strongly deterministic} iff, 
the leftmost contents
$\gamma$ of the memory and the state $q$ completely determine the transition of ${\mathcal A}$,
in particular which letter $b$ (or possibly the empty word) can be read. Therefore, 
such an automaton ${\mathcal A}$ can accept at most one word $w$ from a given configuration.
We say that ${\mathcal A}$ is {\em level-partitioned} iff $\Gamma$ is the disjoint union of subsets $\Gamma_1,\Gamma_2,\ldots,\Gamma_k$
such that, in every transition of ${\mathcal A}$, every occurrence of 
a letter from $\Gamma_i$ is at level $i$.\\
It is easy to transform any $k$-pushdown automaton ${\mathcal A}$ into
another one ${\mathcal A}'$ which recognizes the same language and is level-partitioned.
Moreover, if ${\mathcal A}$ is strongly deterministic then 
${\mathcal A}'$ is strongly deterministic.

\paragraph{Variables}
As a particular case of the defining equation (\ref{alphabet_VA}), we define the set $k-\var(\Gamma)$ of {\em order $k$ variables} over an aphabet $\Gamma$ and a set of states $Q$:
\begin{equation} 
k-\var(\Gamma)=\{( p, \omega, q) \mid p,q \in Q, \omega \in k-\pds(\Gamma) \setminus\{\varepsilon\}\}.
\label{alphabet_kvar}
\end{equation}
We name {\em variable-words} of level $3$, the elements of $(k-\var(\Gamma))^*$ i.e. the words over the alphabet $k-\vterm(\Gamma)$.
\paragraph{Terms}
\label{par-prelimin-terms}
Given a denumerable  alphabet $\Gamma$ of pushdown symbols, we
introduce another alphabet ${\mathcal U}=\{\Omega,\Omega',\Omega'',\ldots,\Omega_1,
\Omega_2,\ldots,\Omega_n,\ldots\}$ of {\em undeterminates}. We suppose
that $\Gamma \cap {\mathcal U} = \emptyset$.
We call a {\em term} of level $k$ over the constant alphabet $\Gamma$
and the alphabet of undeterminates ${\mathcal U}$, any  
$T \in k-\pds(\Gamma \cup {\mathcal U})$ where every occurence of an undeterminate $U$ in $T$
is followed by $[\varepsilon]$, in the rigorous bracketed notation (see example \ref{ex-pushdowns1}, first notation for $\omega$).\\
Such a pushdown $T$, seen as a sequence of planar trees (as in fig.2 p.366 of \cite{Fra-Sen06}), has all
the undeterminates at the leaves, as is the case for classical terms \footnote{However, no arity is attached here to the symbols
of $\Gamma$, unlike what happens for the classical notion of term}.
We denote by $k-\term(\Gamma,{\mathcal U})$ the set of all terms of level $k$ over the constant alphabet $\Gamma$
and the alphabet of undeterminates ${\mathcal U}$.\\
We denote an element of $k-\term(\Gamma,{\mathcal U})$ by
$T[\Omega_1,\Omega_2, \ldots,\Omega_n]$
(resp. $T[\Omega,\Omega',\Omega'']$)
provided that the only undeterminates appearing in $T$ belong to $\{\Omega_1,\Omega_2, \ldots,\Omega_n\}$
(resp. $\{\Omega,\Omega',\Omega''\}$).      
The notation $k-\vterm(\Gamma,{\cal U})$ designates the set of {\em variable-terms} of level $k$, over the pushdown alphabet $\Gamma$ and the set of
undeterminates ${\cal U}$:
\begin{equation} 
k-\vterm(\Gamma,{\cal U})=\{( p, \omega, q) \mid p,q \in Q, \omega \in k-\term(\Gamma,{\cal U}) \setminus\{\varepsilon\}\}.
\label{alphabet_kvterm}
\end{equation}
We name {\em variable-term-words} of level $3$, the elements of $(k-\vterm(\Gamma,{\cal U}))^*$ i.e. the words over the alphabet $k-\vterm(\Gamma,{\cal U})$.
\paragraph{Graded alphabets}
A graded alphabet of height $k$ is an alphabet $\Gamma = \stackrel{.}{\bigcup}_{i \in [1,k]}\Gamma_i$ .
We call the elements of $\Gamma_i$ the letters of {\em level} $i$.
Given $j \in [1,k]$, we call graded pushdown of level $j$ over $\Gamma$, any element $u$ of $j-\pds(\Gamma)$ where every occurrence of a letter $\gamma$ at level $k-j+i$ in $u$ belongs to $\Gamma_{i}$. The notion of graded $k$-variable (over $\Gamma,Q$) is defined in the same way.
Given two disjoint graded alphabets $\Gamma, {\cal U}$, we define in the same way the graded $k$-terms and the graded $k$-variable-terms, 
Given a graded alphabet $\Gamma$ we denote again by $k-\pds(\Gamma), k-\var(\Gamma)$
the corresponding sets of graded pushdowns and variable (hoping that the context will make clear that we use
graded objects). We also denote by $k-\term(\Gamma,{\mathcal U}), k-\vterm(\Gamma,{\mathcal U})$ the sets of graded terms and variable terms.
\begin{example}
 Let $\Gamma=\{ A_1,A_2,A_3,B_1,B_2,B_3,C_1,C_2,C_3,D_1,D_2,D_3\},\;\;
 {\cal U} = \{\Omega_1, \Omega'_1,\Omega_2,\Omega'_2,\Omega_3,\Omega'_3$ be graded alphabets: for every $i \in [1,3]$,
 $$\Gamma_i = \{ A_i,B_i,C_i,D_i\},\;\; {\cal U}_i =\{ \Omega_i,\Omega'_i\}.$$
 Let
 $$T=A_1[A_2[A_3[\varepsilon]\Omega_3[\varepsilon]]B_2[D_3[\varepsilon]C_3[\varepsilon]]]\Omega_1[\varepsilon],\;\;
 T'= A_1[A_1[A_3[\varepsilon]]].\;\;$$
$$T''= A_2[A_3[\varepsilon]B_3[\varepsilon]\Omega_3]\Omega_2, \;\;T'''=A_3[\varepsilon]B_3[\varepsilon]\Omega_3$$
$$U =A_1[\Omega_2[A_3[\varepsilon]\Omega_3]], \;\;U'= A_1[A_2[A_3[\varepsilon]\Omega_2]]$$
$$U''=A_1[A_2[\varepsilon]B_2[\varepsilon]]$$
Here $T$ is a graded $3$-term,\\
$T'$ is not a graded 3-term because $A_1$ occurs at level $2$ in $T$,\\
$T''$ is a graded 2-term,$T'''$ is a graded 1-term.\\
$U$ is not a graded term because $\Omega_2$ occurs at a non-leaf position.\\
$U'$ is not a graded 3-term because $\Omega_2$ occurs at level $3$ in $U'$.\\
$U''$ is a graded 3-term, though it has depth only 2.\\
For the same reasons
$pTq$ is a graded $3$-variable-term while $pT'q$, $pUq$, $pU'q$ are not 3-variable-terms.
\end{example}
\paragraph{Substitutions}
\label{par-substitutions-terms}
Given $T[\Omega_1,\ldots,\Omega_n] \in k-\term(\Gamma,{\cal U})$, and $H_1,\ldots,H_n
\in k'-\term(\Gamma,{\cal U})$, we denote by $T[H_1/\Omega_1,\ldots,H_n/\Omega_n]$ the $(k+k'-1)$-term obtained by substituting
$H_i$ to $\Omega_i$ in $T$. Note that the new $(k+k'-1)$-pushdown thus obtained is really a term.
The map $T \mapsto T[H_1/\Omega_1,\ldots,H_n/\Omega_n]$ is extended as an homomorphism  $(k-\vterm(\Gamma,{\cal U}))^* \rightarrow ((k+k'-1)-\vterm(\Gamma,{\cal U}))^*$ by: for every $T \in k-\term(\Gamma,{\cal U})\setminus\{\varepsilon\}, p,q, \in Q$:
$$ (p,T,q)[H_1/\Omega_1,\ldots,H_n/\Omega_n]:=\;\;(p,T[H_1/\Omega_1,\ldots,H_n/\Omega_n],q).$$
When the alphabets $\Gamma, {\cal U}$ are graded, if the substitution replaces every undeterminate $\Omega_\lambda \in {\cal U}_{j_\lambda}$ by a term $H_\lambda  \in (k-j_\lambda+1)-\term(\Gamma,{\cal U})$, then the result 
$T[H_\lambda/\Omega_\lambda,\lambda \in [1,\Lambda]]$ is again a (graded) $k$-term.
\begin{example}
\label{ex-terms}
Let $\Gamma=\{ A_1,A_2,A_3,B_1,B_2,B_3,C_1,C_2,C_3,D_1,D_2,D_3 \},\;\;
{\cal U} = \{\Omega_1, \Omega'_1,\Omega_2,\Omega'_2,\Omega_3,\Omega'_3\}$.\\
Let $$T=A_1[A_2[A_3[\varepsilon]\Omega_3[\varepsilon]]B_2[D_3[\varepsilon]C_3[\varepsilon]]]\Omega_1[\varepsilon].$$
Using the short notation:
$$T=A_1[A_2[A_3\Omega_3]B_2[D_3C_3]]\Omega_1.$$
Let
$$H_1:= B_1[B_2[\Omega_3]],\;\; H_2 := C_2[A_3B_3\Omega'_3],\;\; H_3 := C_3C_3C_3\Omega_3$$
Then
$$T[H_1/\Omega_1,H_2/\Omega_2,H_3/\Omega_3] = A_1[A_2[A_3C_3C_3C_3\Omega_3]B_2[D_3C_3]]B_1[B_2[\Omega_3]]$$
Let $w := (p,T,q)(q,A_1[\Omega_2],p)$.
Then
$$w[H_1/\Omega_1,H_2/\Omega_2,H_3/\Omega_3] =
(p,A_1[A_2[A_3C_3C_3C_3\Omega_3]B_2[D_3C_3]]B_1[B_2[\Omega_3]],q)(q,A_1[C_2[A_3B_3\Omega'_3]],p)$$

\end{example}
The following ``substitution-principle'' is straightforward and will be widely used in our proofs.
Given some $k$-pda ${\mathcal A}$ over a pushdown alphabet included
in $\Gamma$, we extend the relations $\derivestar{{\mathcal A}}, \computstar{{\mathcal A}}$ to the pushdown alphabet $\Gamma \cup {\mathcal U}$. 
\begin{lemma}
Let $\vec{\Omega}=(\Omega_1,\ldots,\Omega_n),\;\; w \in (k-\vterm(\Gamma,{\cal U}))^*, w'\in \hat{\Gamma \cup {\cal U}})^*$.
If 
$$w \derivedir{{\mathcal A}} w'$$
then\\
1- $w' \in (k-\vterm(\Gamma,{\cal U}))^*$\\
2- for every $\vec{H} \in (k'-\term(\Gamma,{\cal U}))^n$,
$$w [\vec{H}/\vec{\Omega}] \derivedir{{\mathcal A}} w´  [\vec{H}/\vec{\Omega}].$$
These properties still hold for every level-partitionned automaton ${\cal A}$ and graded alphabets $\Gamma,{\cal U}$,
provided that the levels of ${\cal A}$ coincide with the graduations of the alphabets.
\label{substitution_principle}
\end{lemma}
In other words:\\
1- The relation $\derivedir{{\mathcal A}}$ saturates the set $(k-\vterm(\Gamma,{\cal U}))^*$.\\
2- Every substitution preserves the derivation relation.\\
By induction over the integer $m$, it also preserves all the relations $\derive{m}{{\mathcal A}}$ for $m \geq 0$.\\
Point (1) is straighforward. The key-idea for point (2) is that, as $\Gamma \cap {\mathcal U}=
\emptyset$, the symbols $\Omega_i$ can be copied or erased during the
derivation, but they cannot {\em influence} the sequence of rules used in
that derivation.
\paragraph{$k$-computable mappings}
\begin{definition}[$k$-computable mapping]
A mapping $f: A^* \mapsto B^*$ is called {\em $k$-computable}
iff there exists a  strongly deterministic $k$-$pda$ $\mathcal{A}$, over a pushdown-alphabet $\Gamma$
which is level-partitionned, such that $\Gamma$ contains $k-1$ symbols
 $\gamma_1\in \Gamma_1,\ldots,\gamma_i \in \Gamma_i,\ldots,\gamma_{k-1}\in \Gamma_{k-1}$, a state $q_0$, the alphabet $A$ is a subset of  $\Gamma_k$ and for all $w \in A^*$:
$$(q_0,f(w),\gamma_1[\gamma_2 \ldots [\gamma_{k-1}[w]]\ldots]) \computstar{{\mathcal A}} (q_0, \varepsilon,\varepsilon).$$
One denotes by $\mathbb{S}_k(A^*,B^*)$ the set of $k$-computable mappings from $A^*$ to $B^*$.
\label{def_Sk}
\end{definition}
The particular case where $\card(A)=\card(B)=1$ was studied in \cite{Fra-Sen06}.
\subsection{Regular sets of pushdowns of level $k$}
\label{sub_regularlevel-k}
The set $k-\pds(\Gamma)$ of all puhdowns of kevel $k$ over the alphabet $\Gamma$, can be considered as a language over the
following {\em extended alphabet} $\hat{\Gamma}$:\\
$$\hat{\Gamma} := \Gamma \cup \{ x,\bar{x}\},$$
where $x$ (resp. $\bar{x}$) denotes an opening (resp. closing) square bracket.
From this point of view the $3-\pds$ of example \ref{ex-pushdowns1}
$$\omega=A_1[A_2[A_3C_3]B_2[D_3C_3]]B_1[B_2[B_3D_3]].$$
is described by the word
$$A_1xA_2xA_3C_3\bar{x}B_2xD_3C_3\bar{x}\bar{x}B_1xB_2xB_3D_3\bar{x}\bar{x}\in \hat{\Gamma}^*.$$
One can check that, thus turned into a set of words, $k-\pds(\Gamma)$ is a regular subset of $\hat{\Gamma}^*$.
Let ${\mathcal A}=(Q,B,\Gamma,\delta)$ be some $k-\pda$.
For every states $p,q \in Q$, we define the language
$$P_{p,q}({\mathcal A}) := \{ w \in k-\pds(\Gamma)\mid  \exists u \in B^*, (p,w,q) \to_{\mathcal A}^* u\}.$$
\begin{lemma}
  \label{THE-fin-congre}
  Let ${\mathcal A}=(Q,B,\Gamma,\delta)$ be some $k-\pda$.
  There exists a congruence of finite index $\equiv_{\mathcal A}$ over $\hat{\Gamma}^*$ such that:\\
  - $\equiv_{\mathcal A}$ saturates the language $k-\pds(\Gamma)$\\
  - for every $p,q \in Q$, $\equiv_{\mathcal A}$ saturates the language $P_{p,q}({\mathcal A})$.
\end{lemma}
\begin{proof}
  According to Theorem 1, p. 218 of \cite{HagueOng07}, for every pda of order $k$, the set of co-accessible configurations
  (and, more generally, the set of all configurations that are co-accessible from a given regular set of $k$-pushdowns)
  is a regular subset of $k-\pds(\Gamma)$.
  This shows that every set $P_{p,q}({\mathcal A})$ is regular.\\
  Let $\sim$ (resp. $\sim_{{\mathcal A},p,q}$) be the syntactic congruence of the language $k-\pds(\Gamma)$ (resp. of the language $P_{p,q}({\mathcal A})$).
  Let us define
  $$\equiv_{\mathcal A} := \;\;\sim \;\cap \;(\bigcap_{(p,q) \in Q \times Q} \sim_{{\mathcal A},p,q}).$$ 
  Since each binary relation $\sim, \sim_{{\mathcal A},p,q}$ is a congruence of finite index, $\equiv_{\mathcal A}$ is a congruence of finite index
  (over the monoid $\hat{\Gamma}^*$) and it has the required saturation properties.
  \end{proof}
\subsection{Recurrences in monoids}
\label{subsec-rec-monoids}
When considering mappings into {\em words} instead of integers, one is lead to consider
recurrent relations based on the {\em concatenation} operation.
Let us define a more general notion of mapping defined by recurrent relations based on 
the {\em product} operation in some arbitrary monoid $(M,\cdot,1)$.
\begin{definition}[recurrent relations in $M$]
Given a finite  set $I$ and a family of mappings indexed by $I$,
$f_i: A^* \rightarrow M$ (for $i \in I$),
we call {\em system of recurrent relations in $\MM$} over the family $(f_i)_{i \in I}$, 
a system of the form
\begin{equation*}
f_i(aw)= \prod_{j=1}^{\ell(i,a)} f_{\alpha(i,a,j)}(w) \mbox{ for all } i\in I,a \in A, w \in A^*
\label{eq-monoidal-system}
\end{equation*}
where $\ell(a,i) \in \bbbn,\alpha(i,a,j) \in I$ and the symbol $\prod$ stands for 
the extension 
of the binary product in $\MM$ to an arbitrary finite number of arguments.
\label{monoidal_recurrence}
\end{definition}
When the monoid $\MM$ is a finitely generated {\em free monoid}, such a system is called a system of {\em catenative} recurrent relations.
The free monoid is used in the characterisation below of mappings of level $2$, while
the monoid $(\HOM(C^*,C^*),\circ,{\rm Id})$ will be suitable for studying
mappings of level $3$.
As an intermediate tool, we shall use the following more general notion.
\begin{definition}[regular recurrent relations]
Let $\equiv$ be a congruence of finite index on $A^*$. 
Let $\mathcal{C} = X^*/\equiv$, let $I$ be a finite set and $f_i: X^* \rightarrow M$ (for $i \in I$) be a family of mappings.
We call {\em system of regular recurrent relations in $\MM$} over the family $(f_i)_{i \in I}$
a system of the form
$$f_i(aw)= \prod_{j=1}^{\ell(i,a,d)} f_{\alpha(i,a,d,j)}(u_{i,a,d,j}w) \mbox{ for all } i\in I,a \in A, d \in \mathcal{C}, w \in d$$
where $\ell(i,a,d) \in \bbbn,\alpha(i,a,d,j) \in I, u_{i,a,d,j}\in A^*$.
\label{regular_monoidal_recurrence}
\end{definition}
The system is a system of {\em recurrent relations} (i.e. meets the conditions of 
Definition \ref{monoidal_recurrence}) when all the words $u_{i,x,c}$ 
have null length and the congruence $\equiv$ is maximal i.e. has only one class.
The system is called {\em noetherian} when the term-rewriting system consisting 
of all the oriented rules
$$f_i(aw) \rightarrow  \prod_{j=1}^{\ell(i,a,d)} f_{\alpha(i,a,d,j)}(u_{i,a,d,j}w)$$
for all $i\in I,a \in A, d \in \mathcal{C}, w \in d$, is noetherian
(in this rewriting system, both sides are understood as formal terms
\footnote{i.e. both sides of the rules are elements of the free monoid generated by 
the (infinite) set of terms of depth 1 $\{ f_i(w)\mid i \in I, w \in X^*\}$; 
the system is thus a so-called {\em monadic} semi-Thue system over this free monoid)}.
The system is called {\em strict} when all the words $u_{i,x,c}$ are empty.\\
Clearly every strict system is also noetherian.
In particular every system of recurrent relations is a noetherian system of regular recurrent relations, but the converse does not hold in general.
\subsection{Word-mappings of level 2}
\label{subsec-level2}
Let us recall a notion which originated in computational biology, but turns out 
to be useful in general formal language theory (see \cite{Kar-Roz-Sal97}).
\begin{definition}[HDT0L sequences]
Let $f: A^* \rightarrow B^*$. The mapping $f$ is called a HDT0L sequence iff there exists
a finite alphabet $C$,
a homomorphism $H: A^* \rightarrow \HOM(C^*,C^*)$, an homomorhism $h \in \HOM(C^*,B^*)$ and a letter
$c \in C$ such that, for every $w \in A^*$
$$f(w)= h( H^w (c)).$$
\label{hdt0l}
\end{definition}
(here we denote by $H^w$ the image of $w$ by $H$).
The mapping $f$ is called a DT0L when $B=C$ and the homomorphism $h$ is just the identity; $f$ is called a HD0L when $A$ is reduced to one element.\\

The following characterisation of word-mappings of level 2 is proved in 
\cite{Fer-Mar-Sen14}.
\begin{theorem}
Let us consider a mapping $f:A^* \rightarrow B^*$.
The following properties are equivalent:\\
1- $f \in \mathbb{S}_2(A^*,B^*)$\\
2- There exists a finite family $(f_i)_{i\in [1,n]}$ of mappings $A^* \rightarrow B^*$ 
which fulfils a noetherian system of regular recurrent relations of order 1 in $\langle B^*, \cdot,\varepsilon \rangle$ and such that $f=f_1$\\
3- There exists a finite family $(f_i)_{i\in [1,n]}$ of mappings $A^* \rightarrow B^*$ 
which fulfils a strict system of regular recurrent relations 
in $\langle B^*, \cdot,\varepsilon \rangle$ 
and such that $f=f_1$\\
4- $f$ is a HDT0L sequence.\\
5- There exists a finite family $(f_i)_{i\in [1,n]}$ of mappings $A^* \rightarrow B^*$ 
which fulfils a system of catenative recurrent relations and such that $f=f_1$.
\label{characterisation_level2}
\end{theorem}
\begin{proof}
  Lemma 27 of \cite{Fer-Mar-Sen14} proves that:
  $$(1) \Rightarrow (3)\Rightarrow (4) \Rightarrow (5) \Rightarrow (1).$$
  Implication $(3) \Rightarrow (2)$ is obvious.\\
  Let us suppose that (2) holds:\\
$$f_i(aw)= \prod_{j=1}^{\ell(i,a,d)} f_{\alpha(i,a,d,j)}(u_{i,a,d,j}w) \mbox{ for all } i\in I,a \in A, d \in \mathcal{C}, w \in d$$
  where $\ell(i,a,d) \in \bbbn,\alpha(i,a,d,j) \in I, u_{i,a,d,j}\in A^*$.
  Let us consider ${\cal C}= A^*/\equiv$ as a set of undeterminates.
Let us define:
$$R_{i,a}(d):= \prod_{j=1}^{\ell(i,a,d)} f_{\alpha(i,a,d,j)}(u_{i,a,d,j}d) \mbox{ for all } i \in I,a \in A, d \in \mathcal{C}.$$
We denote by $\Sys_{A}$ the rewriting system consisting of the set of rules
$$ f_i(a\cdot w) \rightarrow R_{i,a}(d)[w/d]\;\;\mbox{ for } d \in {\cal C},w \in d.$$
By hypothesis (2), this system ${\cal S}_A$ is noetherian.\\
  Let us denote by $\Sys_{\cal C}$ the rewriting system consisting of the set of rules
  $$ f_i(a \cdot w) \rightarrow R_{i,a}(d)[w/d]\;\;\mbox{ for } d \in {\cal C}, w \in A^*\cdot {\cal C}, w \equiv d,$$
  where $\equiv$ is the self-extension of $\equiv$ to $(A \cup{\cal C})^*$ (defined in \S \ref{par-prelimin-congruences}).
  Let us define, for every $d \in {\cal C}$ , $w_d$ as the smallest element of $d$.
  We claim that, if
  $$f_i(w \cdot d) \rightarrow_{\Sys_{\cal C}}^n  \prod_{j=1}^\ell f_{i_j}(w_j\cdot d)$$
    then
 $$f_i(w \cdot w_d) \rightarrow_{\Sys_{A}}^n \prod_{j=1}^\ell f_{i_j}(w_j\cdot w_d).$$
Hence the rewriting-system $\Sys_{\cal C}$ is noetherian too.
It follows that there exists some irreducible expression $S_{i,a}(d)$ obtained from $R_{i,a}(d)$
by a derivation $\pmod{\Sys_{\cal C}}$.
The new system of relations:
$$f_i(aw)= S_{i,a}(d)[w/d] \mbox{ for all } i\in I,a \in A, d \in \mathcal{C}, w \in d$$
is a {\em strict} system of regular recurrent relations in $\langle B^*, \cdot,\varepsilon \rangle$
which is fulfilled by the family $(f_i)_{i \in I}$.
Hence $(2) \Rightarrow (3)$.\\
\end{proof}

\section{Word-mappings of level 3}
\label{level3}
Let us prove the main result of this paper.

\begin{theorem}
Let us consider a mapping $f:A^* \rightarrow B^*$.
The following properties are equivalent:\\
1- $f \in \mathbb{S}_3(A^*,B^*)$\\
2- There exists a finite family $(H_i)_{i\in [1,n]}$ of mappings $A^* \rightarrow \HOM(C^*,C^*)$ 
which fulfils a system of recurrent relations in $\langle \HOM(C^*,C^*),\circ,{\rm Id}\rangle$,
an element $h \in \HOM(C^*,B^*)$ and a letter $c \in C$ such that, for every $w \in A^*$:
$$f(w)=h(H_1(w)(c)).$$
3- $f$ is a composition of a DT0L sequence $g: A^* \rightarrow C^*$ by a HDT0L sequence $h: C^* \rightarrow B^*$.
\label{characterisation_level3}
\end{theorem}
We prove this theorem through the more technical statement
\begin{theorem}
Let us consider a mapping $f:A^* \rightarrow B^*$.
The following properties are equivalent:\\
1- $f \in \mathbb{S}_3(A^*,B^*)$\\
2- There exists a finite family $(H_i)_{i\in [1,n]}$ of mappings $A^* \rightarrow \HOM(C^*,C^*)$ 
which fulfils a {\em noetherian} system of {\rm regular} recurrent relations in $(\HOM(C^*,C^*),\circ,{\rm Id})$,
an element $h \in \HOM(C^*,B^*)$ and a letter $c \in C$ such that, for every $w \in A^*$:
$$f(w)=h(H_1(w)(c)).$$
3- There exists a finite family $(H_i)_{i\in [1,n]}$ of mappings $A^* \rightarrow \HOM(C^*,C^*)$ 
which fulfils a {\em strict} system of {\rm regular} recurrent relations in $(\HOM(C^*,C^*),\circ,{\rm Id})$,
an element $h \in \HOM(C^*,B^*)$ and a letter $c \in C$ such that, for every $w \in A^*$:
$$f(w)=h(H_1(w)(c)).$$
4- $f$ is a composition of a DT0L sequence $g: A^* \rightarrow C^*$ by a HDT0L sequence $h: C^* \rightarrow B^*$.\\
5- There exists a finite family $(H_i)_{i\in [1,n]}$ of mappings $A^* \rightarrow \HOM(C^*,C^*)$ 
which fulfils a system of recurrent relations in $(\HOM(C^*,C^*),\circ,{\rm Id})$,
an element $h \in \HOM(C^*,B^*)$ and a letter $c \in C$ such that, for every $w \in A^*$:
$$f(w)=h(H_1(w)(c)).$$
\label{techcharacterisation_level3}
\end{theorem}
We first prove that $(1) \Rightarrow (2)$, then we prove that $(2) \Rightarrow (3)\Rightarrow (4)\Rightarrow (5)$
and finally that $(5) \Rightarrow (1)$.\\
Let us consider some $f \in \mathbb{S}_3(A^*,B^*)$ and 
let us show that it fulfils condition (2).\\
Let $\mathcal{A}=(Q,B,\Gamma,\delta,q_0,Z_0)$ be a strongly deterministic 
3-iterated pushdown automaton, such that 
$\mathcal{A} \subseteq \Gamma$ and  $\mathcal{A}$ computes $f$.
We can normalize the automaton $\mathcal{A}$ in such a way that:\\
\begin{description}
\item{(LP)} $\mathcal{A}$ is level-partitioned (see \S \ref{par-normal-aut})
\item{(RL)} the only transitions that read a letter from $B$  are of the form:\\
$$ \delta(p,b,S)=(q,pop_1)$$
for some $p,q \in Q, b \in B, S \in \Gamma_1$.
\item{(PI)} the only push intructions occuring in $\delta$ are of the form 
$\push_j(\gamma)$ with $\gamma \in \Gamma^2$.
\end{description}
Let us define from this automaton a finite alphabet $\hat{\cal W}$ and a family of mappings $H_i: \Gamma^* \rightarrow \HOM(\hat{\cal W}^*,\hat{\cal W}^*)$ that is closely related to the computations of $\mathcal{A}$ and fulfils a
{\em noetherian system of regular recurrent relations of order $1$ in $\HOM(\hat{\cal W}^*,\hat{\cal W}^*)$}.\\
\paragraph{General idea}
\label{par-genral-idea}
The general idea that guides our proof of $(1) \Rightarrow (2)$ in Theorem \ref{techcharacterisation_level3} is as follows.
Every $2-\pds$ atom $T[w]$ (where $T$ has level 2 and $w$ is a product of symbols of level 3) {\em acts} on the set of $3-\pds$ by:
$$ S[u] \mapsto S[T[w]u].$$
We represent all the possible $S[u]$ by a finite alphabet ${\cal W}$, in such a way that the family of actions of $T[w]$ (for $T \in \Gamma_2$) is represented by a single homomorphism $H_i^w: {\cal W}^* \rightarrow {\cal W}^*$.
The final outcome is that, to every word $w \in \Gamma_3^*$ is associated a finite family of homomorphisms $H_i^w: {\cal W}^* \rightarrow {\cal W}^*$ with $i \in I$ , where $I$ is a finite set of indices that encode the various letters $T \in \Gamma_2$ and various subsets
of ${\cal W}$ on which the action is restricted ( with additional homomorphisms that do not depend on the argument $w$ but
ease the definition of a compositional recurrence that links the $H^{aw}_i$ with the $H^{w}_i$ for letters $a \in \Gamma_3$,
see the defining equations (\ref{eq-technical-homo1},\ref{eq-technical-homo2})).
\subsection{Alphabets $E$ and  ${\cal W}$}
\label{sub-alphabetW}
\paragraph{Congruences}
Let $\equiv_{\mathcal{A}}$ be a finite-index congruence over $\hat{\Gamma}^*$ given by Lemma \ref{THE-fin-congre} of subsection \ref{sub_regularlevel-k}.
For every integer $j \in [1,3]$ , we denote by $(\equiv_{\mathcal{A},j})$ the restriction of $\equiv_{\mathcal{A}}$ over $(4-j)-\pds(\Gamma)$:
$$\forall w,w'\in (4-j)-\pds(\Gamma),\;\;w \equiv_{\mathcal{A},j} w' \Leftrightarrow w \equiv_{\mathcal{A}} w'.$$
The following facts are immediate:\\
\begin{eqnarray}
\mbox{if } w \equiv_{\mathcal{A},1} w',& \mbox{then } \forall p,q \in Q, \Language(\mathcal{A},pwq)=\emptyset \Leftrightarrow \Language(\mathcal{A},pw'q)=\emptyset,&\label{eq-decrease-level}\\\nonumber
\mbox{if } w \equiv_{\mathcal{A},2} w',& \mbox{then } \forall S \in \Gamma_1, S[w] \equiv_{\mathcal{A},1} S[w'],&\\\nonumber
\mbox{if } w \equiv_{\mathcal{A},3} w',& \mbox{then } \forall T \in \Gamma_2, T[w] \equiv_{\mathcal{A},2} T[w'].&\label{eq-three-facts}
\end{eqnarray}
From now on, we abbreviate each congruence $\equiv_{\mathcal{A},i}$ by $\equiv_i$.\\
\paragraph{Alphabet $E$}
Let
$$E := 2-\pds(\Gamma)/\equiv_2.$$
In the sequel we use $E$ as a graded alphabet of undeterminates, where the level of all the symbols is 2.

\paragraph{Alphabet ${\cal W}$}
We define the finite alphabet
$$
{\cal W} := \{ (p S[e] q) \mid p, q \in Q, S \in \Gamma_1, e \in 2-\pds(\Gamma)/\equiv{2} \}
$$
Note that ${\cal W}$ is a finite subset of the infinite alphabet $3-\vterm(\Gamma,E)$.
\subsection{Morphisms $H_i^w$}
We use letters $S, S_1,\ldots,S_n,\ldots$ to denote elements of $\Gamma_1$,
$T, T_1,\ldots,T_n,\ldots$ to denote elements of $\Gamma_2$,
$a, a_1,\ldots,a_n,\ldots$ to denote elements of $\Gamma_3$,
$\Omega, \Omega_1,\ldots,\Omega_n,\ldots$ to denote undeterminates other than those in $E$ i.e. symbols which are not elements of $\Gamma$ (see \S \ref{par-prelimin-terms} of section \ref{sec-prelimin2}).
Let 
$$I_0 := {\cal P}({\cal W}) \times \Gamma_2.$$
Let $i \in I_0$: it has  the form $i=({\cal V}, T)$ where 
${\cal V} \subseteq {\cal W}, T \in \Gamma_2$. We also denote by 
$\Al(i)$ the first component ${\cal V}$ of $i$.
For every $w \in \Gamma_3^*$, we define the homomorphism $H_i^w$ by:\\
for every letter  $W=(p S[e] q) \in {\cal W}$, if
\begin{equation}
W \in {\cal V}     \label{alphabet_condition}
\end{equation}
and
\begin{equation}
(p S[T[w] \Omega]q) \rightarrow_{\mathcal{A}}^+ \prod_{j=1}^{\ell(i,W)} (p_{i,j}S_{i,j}[\Omega] q_{i,j})\label{derivational_condition}
\end{equation}
and
\begin{equation}
\forall j \in [1, \ell(i,W)],\forall t \in e,\;\; \Language(p_{i,j}S_{i,j}[t] q_{i,j}) \neq \emptyset
\label{nonemptiness_condition}
\end{equation}
then
\begin{equation}
H_i^w: \;\;(p S[e] q) \mapsto  \prod_{j=1}^{\ell(i,W)} (p_{i,j}S_{i,j}[e] q_{i,j}).
\end{equation}
otherwise
\begin{equation}
H_i^w:\;\;(p S[e] q) \mapsto  (p S[e] q).
\end{equation}
We denote by $H_i$ the mapping: $\Gamma^* \rightarrow \HOM({\cal W}^*,{\cal W}^*)$ defined by
$w \mapsto H_i^w$. 
\begin{lemma}
For every $i= ({\cal V}, T)\in I_0$, word $w \in \Gamma^*$  and letter $W \in {\cal V}$,  there exists at most one possible righthandside for the equation (\ref{derivational_condition}) that also satisfies equation 
(\ref{nonemptiness_condition}).
\end{lemma}
For every $i \in I_0$, we denote by $\Al(i,w)$ the set of letters 
$W \in \Al(i)$, such that there exists exactly one righthandside for the equation (\ref{derivational_condition}) that satisfies also  equation(\ref{nonemptiness_condition}).
Let us remark that, if $w \in \Gamma^* \setminus \Gamma_3^*$, then $\Al(i,w)= \emptyset$.
For this reason most of our further statements are trivial for words in 
$\Gamma^* \setminus \Gamma_3^*$.
\begin{lemma}
For every $w,w'\in \Gamma_3^*$, if $w \equiv_{3} w'$, then, for every $ i \in I_0$,
$$\Al(i,w)= \Al(i,w').$$
\label{l-equivA3_determines_ALiw}
\end{lemma}
\begin{proof}
Suppose that $i=({\cal V},T)$. Let $W=pS[e]q$ and suppose that
\begin{equation} 
W \in {\cal V}.
\label{W_inV}
\end{equation}
We observe that $pS[e]q \in \Al(i,w)$ iff 
\begin{equation}
\forall u \in e,\;\;\Language(\mathcal{A},pS[T[w]u]q) \neq \emptyset
\label{CNS_alpha(i,w)}
\end{equation}
Let $w,w' \in \Gamma_3^*$ such that $w \equiv_{3} w'$. By the three facts (\ref{eq-three-facts}) of subsection \ref{sub-alphabetW}, for every $u \in e$:
$$T[w]u \equiv_{2} T[w']u$$
hence\\ 
$$S[T[w]u] \equiv{1} S[T[w']u],$$ hence\\
$$\Language(\mathcal{A},pS[T[w]u]q ) = \emptyset \Leftrightarrow
\Language(\mathcal{A},pS[T[w']u]q) = \emptyset.$$
Using observation (\ref{CNS_alpha(i,w)}) we can conclude that, under the hypothesis that 
$(W \in {\cal V}\;\; \& \;\;w \equiv_{3} w')$ it is true that
\begin{equation}
W \in \Al(i,w) \Leftrightarrow W \in \Al(i,w').
\label{eq-same-alphabets}
\end{equation}
Under the hypothesis that $W \notin {\cal V}$,
we obtain that $W \notin  \Al(i,w) \& W \notin  \Al(i,w')$ which implies again that
\begin{equation}
W \in \Al(i,w) \Leftrightarrow W \in \Al(i,w').\
\label{eq-same-alphabets2}
\end{equation}
By (\ref{eq-same-alphabets},\ref{eq-same-alphabets2}) the lemma is proved.
\end{proof}
Owing to Lemma \ref{l-equivA3_determines_ALiw}, for every $ i \in I_0$ and every class $d \in \Gamma_3^*/\equiv_3$,  we denote by $\Al(i,d)$ the set $\Al(i,w)$ for any $w \in d$.

\subsection{Mixed recurrence}
\begin{lemma}
Let $i \in I_0\; (i=({\cal V},T)), d \in 1-\pds(\Gamma)/\equiv_3, a \in \Gamma_3, 
w \in 1-\pds(\Gamma_3)$  and $W \in {\cal W}\; (W=pS[e]q)$. 
There exists an integer $\ell(i,a,d,W) \in [1,2]$ and 
for every $j \in [1,\ell(i,a,d,W)]$ there exist 
indices $\alpha(i,a,d,W,j),\beta(i,a,d,W,j) \in I_0$, 
words $\omega_{i,a,d,W,j} \in \Gamma_3^{\leq 2}$ ,
letters $V_{i,a,d,W,j} \in {\cal W}$ and an alphabetic homomorphism $\Phi_{i,a,d,W,j}
\in \HOM({\cal W}^*,{\cal W}^*)$ such that\\
\begin{equation}
w \in d \Rightarrow H_i^{aw}(W)= \prod_{j=1}^{\ell(i,a,d,W)}
H_{\alpha(i,a,d,W,j)}^{\omega_{i,a,d,W,j} \cdot w} \circ \Phi_{i,a,d,W,j} 
\circ H_{\beta(i,a,d,W,j)}^{\omega_{i,a,d,W,j} \cdot w}(V_{i,a,d,W,j})
\label{mixed_recurrence}
\end{equation}
and one of the following cases occurs:\\
{\bf Case 1.1:}
$\ell(i,a,d,W)=1, \alpha(i,a,d,W,1)=(\emptyset,T),\beta(i,a,d,W,1)=(\emptyset,T)$. 
{\bf Case 1.2:}
$\;\ell(i,a,d,W)=1, \alpha(i,a,d,W,1)=(\{V_{i,a,d,W,1}\},T'),
\Phi_{i,a,d,W,1} = {\rm Id}_{{\cal W}^*},
\beta(i,a,d,W,1)=(\emptyset,T)$ and
$$ W[T[aw]e/e]\rightarrow_{\mathcal{A}}V_{i,a,d,W,1}[T'[\omega_{i,a,d,W,j}w]e/e] \rightarrow^*_\mathcal{A} H_i^{aw}(W).$$
{\bf Case 1.3:}
$\ell(i,a,d,W)=1, \alpha(i,a,d,W,1)=(\{V_{i,a,d,W,1}\},T'),\\
e'=[T''[\omega_{i,a,d,W,j}\cdot w]]_{\equiv{\mathcal{A},2}}\cdot e,$
$$\Phi_{i,a,d,W,1}: p'S'[e'']q' \mapsto p'S'[e'']q'\mbox{ if } e''\neq e',\;\;
p'S'[e']q' \mapsto p'S'[e]q',$$
$\beta(i,a,d,W,1)= (\Al(\Phi_{i,a,d,W,1}(H_{\alpha(i,a,d,W,1)}^{\omega_{i,a,d,W,1} \cdot w} (V_{i,a,d,W,1}))),T'')$,
$V_{i,a,d,W,1}=p'S[e']q$ for some $p' \in Q$, and
$$ W[T[aw]e/e]\rightarrow_{\mathcal{A}}
V_{i,a,d,W,1}[T'[\omega_{i,a,d,W,1}w]T''[\omega_{i,a,d,W,1}w]e/e']\rightarrow^*_\mathcal{A} H_i^{aw}(W)$$
{\bf Case 2:}
$\ell(i,a,d,W)=2, \omega_{i,a,d,W,1}= \omega_{i,a,d,W,2}=a,\alpha(i,a,d,W,j)=(\{V_{i,a,d,W,j}\},T),
\Phi_{i,a,d,W,j}={\rm Id}_{{\cal W}^*},
\beta(i,a,d,W,j)= (\emptyset,T)$ and
$$ W[T[aw]e/e]\rightarrow_{\mathcal{A}}^+
V_{i,a,d,W,1}[T[aw]e/e]\cdot V_{i,a,d,W,2}[T[aw]e/e] \rightarrow^*_\mathcal{A} H_i^{aw}(W).$$
\label{l-mixed_recurrence_exists}
\end{lemma}
Let us sketch a proof of this lemma.
In the case of an operation $push_1$,the value of $\ell(i,a,d,W)$ is 2 while the case $\pop_1$ is impossible and for the four remaining cases ($\pop_2,\pop_3,\push_2,\push_3$)
the value of $\ell(i,a,d,W)$ is 1.
In the case of an operation $push_1$, the two letters
$V_{i,a,d,W,1},V_{i,a,d,W,2}$ depend on an intermediate state which is determined by the congruence class $d$.\\
In the case of an operation $push_2$, the composition of two non-trivial homomorphisms  is really
required in the righthandside. In the other cases we can choose $\Phi_{i,a,d,W,j}:= {\rm Id}_{{\cal W}^*}$ and $\beta(i,a,d,W,j):=(\emptyset,T)$,
so that, for every $u$, $H_{\beta(i,a,d,W,j)}^{u}={\rm Id}_{{\cal W}^*}$.\\
Let us proceed now to the formal proof.
\begin{proof}
For $W \notin \Al(i,ad)$\footnote{ i.e. with full rigor $\Al(i,[a]_{\equiv_3}\cdot d)$} we just choose $\ell(i,a,d,W):=1$, $\omega(i,a,d,W,1):= \varepsilon$
and $\alpha(i,a,d,W,1):= (\emptyset, T), \Phi_{i,a,d,W,1}:= {\rm Id}_{{\cal W}^*},\beta(i,a,d,W,1):= (\emptyset, T),V_{i,a,d,W,1}:= W$.
These choices ensure that
$$
H_{\alpha(i,a,d,W,j)}^{\omega_{i,a,d,W,j} \cdot w} \circ \Phi_{i,a,d,W,1} \circ H_{\beta(i,a,d,W,j)}^{\omega_{i,a,d,W,j} \cdot w}= {\rm Id}_{{\cal W}^*}.
$$
Equation (\ref{mixed_recurrence}) thus holds and 
Case 1.1. is realized.\\
Let us now suppose that $W \in \Al(i,ad)$.
Since $\mathcal{A}$ is strongly deterministic, 
there is at most one  element $\bar{b} \in B \cup \{\varepsilon\}$ such that a transition $\delta(p,\bar{b},STa)$  is defined. The fact $W \in \Al(i,ad)$ implies that $\Language(\mathcal{A}, pS[T[aw]u]q) \neq \emptyset$
for every $u \in e$, hence that there exists exactly one such transition.
At last, the restriction (RL) that is assumed for $\mathcal{A}$ shows that this transition has the form:$$\delta(p,\varepsilon,STa)=(p_1,op) \mbox{ where } op \in PUSH(\Gamma)\cup \{\pop_2,\pop_3\}
$$ 
We distinguish several cases corresponding to these five possible values of $op$.\\
\underline{{\bf $pop_2$}}:$pS[T[aw]e]q \rightarrow p_1S[e]q$.\\
Let us choose: $\ell(i,a,d,W):=1,\alpha(i,a,d,W,j):=(\emptyset,T),\Phi_{i,a,d,W,j}:={\rm Id}_{{\cal W}^*},
\beta(i,a,d,W,j):=(\emptyset,T),\omega_{i,a,d,W,1}:= \varepsilon, V_{i,a,d,W,1}:= p_1S[e]q$.\\ 
These choices ensure that the value of both sides of (\ref{mixed_recurrence}) is $p_1S[e]q$ and that 
Case 1.1 is realized.\\ 
\underline{{\bf $pop_3$}}:$pS[T[aw]e]q \rightarrow p_1S[T[w]e]q$.\\
Let us choose: $\ell(i,a,d,W):=1,\alpha(i,a,d,W,j):=(\emptyset,T),\Phi_{i,a,d,W,j}:={\rm Id}_{{\cal W}^*}, 
\beta(i,a,d,W,j):=(\emptyset,T), \omega_{i,a,d,W,1}:= \varepsilon, V_{i,a,d,W,1}:= p_1S[e]q.$\\ 
These choices ensure that equation (\ref{mixed_recurrence}) holds and that Case 1.2 is realized.\\ 
\underline{{\bf $push_1(S_1S_2)$}}:\\
Since $\Language(\mathcal{A}, pS[T[aw]u]q) \neq \emptyset$ (for a given $u \in e$) and 
$pS[T[aw]u]q \rightarrow_\mathcal{A} p_1S_1[T[aw]e]S_2[T[aw]e]q$, there exists a unique $r \in Q$
such that $\Language(\mathcal{A}, p_1S_1[T[aw]u]r)~\neq~\emptyset$ and
$\Language(\mathcal{A}, rS_2[T[aw]u]q)~\neq~\emptyset$.
This state $r$ depends on $d,e$ only and not really on $w,u$ (by definition of the equivalences
$\equiv{j}$ for $j \in [1,3]$ and by the property that they are congruences).
We also have
$$pS[T[aw]u]q \rightarrow_\mathcal{A}^2 (p_1S_1[T[aw]e]r)(rS_2[T[aw]e]q).$$
Let us choose $\ell(i,a,d,W):=2,V_{i,a,d,W,1}:= p_1S_1[e]r, V_{i,a,d,W,2}:= rS_2[e]q,
\alpha(i,a,d,W,j):=(\{V_{i,a,d,W,j}\},T),\Phi_{i,a,d,W,j}:={\rm Id}_{{\cal W}^*}, 
\beta(i,a,d,W,j):=(\emptyset,T), \omega_{i,a,d,W,1}:= \omega_{i,a,d,W,2}:=a.$
Equation (\ref{mixed_recurrence}) holds and Case 2 is realized.\\
\underline{{\bf $push_2(T'T'')$}}:In this case
\begin{equation}pS[T[aw]e]q \rightarrow p_1S[T'[aw]T''[aw]e]q.
\label{e-firststep_push2}
\end{equation}
Let us choose: $\ell(i,a,d,W):=1, e':= [T[aw]]_{\equiv{2}}\cdot e,V_{i,a,d,W,1}:= p_1S[e']q,
\alpha(i,a,d,W,1):=(\{p_1S[e']q\},T'), \Phi_{i,a,d,W,1}:p'S'[e'']q' \mapsto p'S'[e'']q'\mbox{ if } e''\neq e',\;\;
p'S'[e']q' \mapsto p'S'[e]q',\;\;
\beta(i,a,d,W,1):=(\Al(\Phi_{i,a,d,W,1}(H_{\alpha(i,a,d,W,1)}^{aw}(p_1S[e']q))),T''),\omega_{i,a,d,W,1}:= a$.\\
Nota: the action of $\Phi_{i,a,d,W,1}$ on letter $p'S[e'']q'$ with $e''\neq e'$ has no influence on equation (\ref{mixed_recurrence}).\\
In order to lighten the notation we abbreviate $\alpha(i,a,d,W,1)$ as $\alpha$, $\beta(i,a,d,W,1)$
as $\beta$ and $\Phi_{i,a,d,W,1}$ as $\Phi$ in the rest of this proof.
Let us show that equation (\ref{mixed_recurrence}) holds.
By definition of the mappings $H_i$ we have:
$$p_1S[T'[aw]e']q \rightarrow^*_\mathcal{A} H_\alpha^{aw}(p_1S[e']q).$$
Applying the substitution $[T''[aw]e/e']$ on both sides of this derivation we obtain
\begin{equation}
p_1S[T'[aw]T''[aw]e]q \rightarrow^*_\mathcal{A} (H_\alpha^{aw}(p_1S[e']q))[T''[aw]e/e'].
\label{e-secondstep_push2}
\end{equation}
Let us consider an arbitrary letter $V$ occuring in the word $H_\alpha^{aw}(p_1S[e']q)$:
it has the form $V=p'S'[e']q'$ for some $p',q' \in Q, S' \in \Gamma_1$.
By definition of the mappings $H_i$ we have:
$$p'S'[T''[aw]e]q' \rightarrow^*_\mathcal{A} H_\beta^{aw}(p'S'[e]q').$$
This derivation can also be written as:
$$V[T''[aw]e/e'] \rightarrow^*_\mathcal{A} H_\beta^{aw}(V[e/e'])=
H_\beta^{aw}(\Phi(V)).$$
Since the above derivation is valid for every letter $V$ occuring in $H_\alpha^{aw}(p_1S[e']q)$, we conclude that
\begin{equation}
H_\alpha^{aw}(p_1S[e']q)[T''[aw]e/e'] \rightarrow^*_\mathcal{A} 
H_\beta^{aw}(\Phi(H_\alpha^{aw}(p_1S[e']q))).
\label{e-thirdstep_push2}
\end{equation}
The product of the three derivations (\ref{e-firststep_push2}),(\ref{e-secondstep_push2}) and (\ref{e-thirdstep_push2}) gives a derivation
\begin{equation}
pS[T[aw]e]q   \rightarrow^*_\mathcal{A} 
H_\beta^{aw}(\Phi(H_\alpha^{aw}(p_1S[e']q))).
\label{e-allsteps_push2}
\end{equation}
Let us check now that the extremity of derivation (\ref{e-allsteps_push2}) is nothing else than
$H_i^{aw}(pS[e]q)$. We introduce, for every congruence class 
$f \in 2-\pds(\Gamma)/{\equiv{2}}$ the subalphabets:
$${\cal W}_f:= \{ p'S'[f]q' \mid p',q' \in Q, S' \in  \Gamma_1 \},\;\;$$
$${\cal W}^0_f:= \{ p'S'[f]q' \mid p',q' \in Q, S' \in  \Gamma_1,\forall u \in f,
\Language(\mathcal{A}, p'S'[u]q') \neq \emptyset  \}.$$
Since $pS[e]q \in \Al(i,aw)$, we know that
$$ \forall u \in e, \Language(\mathcal{A}, pS[T[aw]u]q) \neq \emptyset.$$
Every derivation (modulo $\rightarrow_\mathcal{A}$) starting from $pS[T[aw]u]q$
has, as its first step, derivation (\ref{e-firststep_push2}). Hence
$$ \forall u \in e, \Language(\mathcal{A}, p_1S[T'[aw]T''[aw]u]q) \neq \emptyset.$$
By definition of $e'$, $\forall u \in e, T''[aw]u \in e'$. This ensures that,
$$ \forall u' \in e', \Language(\mathcal{A}, p_1S[T'[aw]u']q) \neq \emptyset.$$
Since $\alpha=(\{p_1S[e']q\},T')$, the above non-emptiness assertion shows that
$$p_1S[e']q \in \Al(\alpha,aw).$$
By definition of $H_\alpha^{aw}$ we have
\begin{equation}
H_\alpha^{aw}(p_1S[e']q ) \in ({\cal W}_{e'}^0)^*.
\label{e-letters_Ha}
\end{equation}
Let $V\in {\cal W}$ be  factor of $H_\alpha^{aw}(p_1S[e']q)$.
By (\ref{e-letters_Ha}) $V \in {\cal W}_{e'}^0$, in particular it has the form
$V = p'S'[e']q'$.
Since $\forall u \in e, T''[aw]u \in e'$, we then have
$$ \forall u \in e, \Language(\mathcal{A}, p'S'[T''[aw]u]q') \neq \emptyset,$$
hence $p'S'[e]q' \in \Al(\beta, aw)$, or, equivalently,
$\Phi(V) \in \Al(\beta, aw)$.
It follows that
$$\Phi(H_\alpha^{aw}(p_1S[e']q ) \in \Al(\beta, aw)^*$$
hence
\begin{equation}
H_\beta^{aw}(\Phi(H_\alpha^{aw}(p_1S[e']q )))) \in ({\cal W}_{e}^0)^*.
\label{e-letters_HbPha}
\end{equation}
By (\ref{e-allsteps_push2}) and (\ref{e-letters_HbPha}), 
$$H_i^{aw}(pS[e]q)= H_\beta^{aw}(\Phi(H_\alpha^{aw}(p_1S[e']q ))))$$i.e. equation (\ref{mixed_recurrence}) is valid.
The conditions of Case 1.3 are also valid.\\
\underline{{\bf $push_3(\omega)$}}:$pS[T[aw]e]q \rightarrow p_1S[T[\omega\cdot w]e]q$.\\
Let us choose: $\ell(i,a,d,W):=1, \;V_{i,a,d,W,1}:= p_1S[e]q,\; \omega(i,a,d,W,1):= \omega,
\;\alpha(i,a,d,W,1):= (\{p_1S[e]q\}, T),\; \Phi_{i,a,d,W,1}:= {\rm Id}_{{\cal W}^*},
\;\beta(i,a,d,W,1):= (\emptyset, T).$
These choices ensure that equation (\ref{mixed_recurrence}) holds and 
Case 1.2 is realized.

\end{proof}
Note that the above recurrence relations (\ref{mixed_recurrence}) involve both the product operation
(over ${\cal W}^*$) and the composition operation (over $\HOM({\cal W}^*,{\cal W}^*)$).
Therefore we call them  the {\em mixed} recurrence relations. We shall transform these mixed recurrence relations into ``purely'' compositional recurrence relations in next paragraph.
\subsection{Compositional recurrence}
Let us define a new alphabet ${\hat{\cal W}}$ extending ${\cal W}$ and a family of homomorphims ${\hat{\cal W}}^* \rightarrow {\hat{\cal W}}^*$ that satisfy a system of compositional recurrence relations.
Let
$${\hat{\cal W}}:= {\cal W} \cup ({\cal W} \times {\cal W} \times \{1,2\}) \cup \{X\},$$
where $X$ is a new letter not in ${\cal W} \cup ({\cal W} \times {\cal W} \times \{1,2\})$.
For every $i,a,d$ we introduce the following homomorphisms ${\hat{\cal W}}^* \rightarrow {\hat{\cal W}}^*$:\\
$\psi_{i,a,d}$:
\begin{eqnarray}
X \mapsto & X & \nonumber\\
W \mapsto &(W,W,1) & \mbox{ if } W \in {\cal W} \setminus \Al(i,ad)\nonumber\\
W \mapsto &(V_{i,a,d,W,1},W,1) & \mbox{ if } W \in \Al(i,ad) \mbox{ and } \ell(i,a,d,W)=1\nonumber\\
W \mapsto &(V_{i,a,d,W,1},W,1)\cdot (V_{i,a,d,W,2},W,2) 
& \mbox{ if } W \in \Al(i,ad)\mbox{ and } \ell(i,a,d,W)=2\nonumber\\
(V,W,j) \mapsto & X & \mbox{ if } V,W \in {\cal W}, j \in \{1,2\}
\end{eqnarray}
$\hat{\psi}_{i,a,d}$:
\begin{eqnarray}
X \mapsto & X & \nonumber\\
V \mapsto & X & \mbox{ if } V \in {\cal W}\nonumber\\
(V,W,j) \mapsto & V & \mbox{ if } V,W \in {\cal W}, j \in \{1,2\}
\end{eqnarray}
For every $i,a,d,W,j$ we introduce the following homomorphisms ${\hat{\cal W}}^* \rightarrow {\hat{\cal W}}^*$:\\
$\theta_{i,a,d,W,j}$:
\begin{eqnarray}
X \mapsto & X & \nonumber\\
V \mapsto & X& \mbox{ if } V \in {\cal W} \nonumber\\
(V,W,j) \mapsto & V & \mbox{ if } V \in {\cal W} \nonumber\\
(V,W',j') \mapsto &(V,W',j') & \mbox{ if } V, W' \in {\cal W}, j' \in \{1,2\}, 
W'\neq W \mbox{ or } j' \neq j.
\end{eqnarray}
$\hat{\theta}_{i,a,d,W,j}$:
\begin{eqnarray}
X \mapsto & X & \nonumber\\
V \mapsto & (V,W,j)& \mbox{ if } V \in {\cal W} \nonumber\\
(V,V',j') \mapsto & (V,V',j') & \mbox{ if } V,V' \in {\cal W}, j' \in \{1,2\}
\end{eqnarray}

We extend the homomorphisms $H_i^w$ to $\hat{{\cal W}}^*$ by setting:\\
\begin{eqnarray}
X & \mapsto  & X \nonumber\\
(V,V',j') \mapsto & X & \mbox{ if } V,V' \in {\cal W}, j' \in \{1,2\}
\end{eqnarray}
As well, we extend the homomorphisms $\Phi_{i,a,d,W,j}$ to $\hat{{\cal W}}^*$ by setting:\\
\begin{eqnarray}
X & \mapsto  & X \nonumber\\
(V,V',j') \mapsto & X & \mbox{ if } V,V' \in {\cal W}, j' \in \{1,2\}
\end{eqnarray}
\begin{lemma}
For every $i \in I_0, d \in 1-\pds(\Gamma)/\equiv{3}, a \in \Gamma, 
w \in 1-\pds(\Gamma)$ and every letter $W \in \Al(i,ad)$, 
there exists an integer $\ell(i,a,d,W) \in [1,2]$ and 
for every $j \in [1,\ell(i,a,d,W)]$ there exist 
indices $\alpha(i,a,d,W,j),\beta(i,a,d,W,j) \in I_0$ and 
words $\omega_{i,a,d,W,j} \in \Gamma_3^{\leq 2}$ 
such that\\
\begin{eqnarray}
w \in d &\Rightarrow &\nonumber\\
H_i^{aw}&=& \psi_{i,a,d} \circ \prod_{W \in \Al(i,ad)} 
(\prod_{j=1}^{\ell(i,a,d,W)}
\theta_{i,a,d,W,j}\circ 
H_{\alpha(i,a,d,W,j)}^{\omega_{i,a,d,W,j} \cdot w} \circ \Phi_{i,a,d,W,j}
\circ H_{\beta(i,a,d,W,j)}^{\omega_{i,a,d,W,j} \cdot w}
\circ \hat{\theta}_{i,a,d,W,j})\nonumber\\
&&\circ \hat{\psi}_{i,a,d}.
\label{compositionnal_recurrence1}
\end{eqnarray}
\end{lemma}
In the above lemma the symbols $\prod$ stand for the extension of the composition law to an arbitrary finite number of arguments.
This law is non-commutative, so that the ordering of the five homomorphisms on the righthand-side is significant;
however, for the particular operands of the symbols $\prod$ , the result does {\em not} depend on their ordering).
\begin{proof}
For a letter $Y \in \hat{\cal W} \setminus {\cal W}$ , both sides of equation (\ref{compositionnal_recurrence1}) map $Y$ on $X$.\\
For a letter $Y \in {\cal W}$, one can check that both sides of equation (\ref{compositionnal_recurrence1}) map $Y$ on the same word, by using equations (\ref{mixed_recurrence}) and the definitions of $\psi_*,\hat{\psi}_*,\theta_*,\hat{\theta}_*$.
\end{proof}
We define a larger family of mappings by setting:
$$I := I_0 \cup \{(i,a,d,s) \mid s \in \{-1,+1\}\}
\cup \{((i,a,d,W,j,s) \mid s \in \{-1,0,+1\}\}.$$
The homomorphism $H_i^w$ is already defined for $w \in \Gamma^*, i \in I_0$.
We now define:
\begin{equation}
H^w_{i,a,d,1}= \psi_{i,a,d},\;\;
H^w_{i,a,d,-1}= \hat{\psi}_{i,a,d},\;\;
\label{eq-technical-homo1}
\end{equation}
\begin{equation}
H^w_{i,a,d,W,j,1}= \theta_{i,a,d,W,j},\;\;
H^w_{i,a,d,W,j,0}= \Phi_{i,a,d,W,j},\;\;
H^w_{i,a,d,W,j,-1}= \hat{\theta}_{i,a,d,W,j}.
\label{eq-technical-homo2}
\end{equation}
and we extend the map $\Al(*,*)$ which is already defined for arguments $(i,w) \in I_0 \times \Gamma_3^*$ by setting
$$ \Al(i,w):= \emptyset\mbox{ for } i \in I \setminus I_0, w \in \Gamma_3^*.$$
With these new notations, relations (\ref{compositionnal_recurrence1}) can be rewritten as:
for every $w \in d \;\;\& \;\;i\in I_0$
\begin{eqnarray}
H_i^{aw}&=& H^w_{i,a,d,1} \nonumber\\
&&\circ \prod_{W \in \Al(i,ad)} 
(\prod_{j=1}^{\ell(i,a,d,W)}
H^w_{i,a,d,W,j,1}\circ 
H_{\alpha(i,a,d,W,j)}^{\omega_{i,a,d,W,j} \cdot w} 
\circ H^w_{i,a,d,W,j,0}
\circ H_{\beta(i,a,d,W,j)}^{\omega_{i,a,d,W,j} \cdot w}
\circ H^w_{i,a,d,W,j,-1})\nonumber\\
&&\circ H^w_{i,a,d,-1}.
\label{compositionnal_recurrence21}
\end{eqnarray}
while it is obviously true that: for every $w \in d \;\;\& \;\;i\in I \setminus I_0$
\begin{equation}
H_i^{aw} = H_i^{w}. 
\label{compositionnal_recurrence22}
\end{equation}
Thus the family $(H_i)_{i \in I}$ fulfills a system of regular recurrent relations 
in the monoid $\HOM({\cal W}^*,{\cal W}^*)$. 
The aim of next paragraph is to prove that this system is {\em noetherian}.
\subsection{Termination}
\label{subsub-termination}
Let us denote by ${\cal C}_1$ the rewriting system consisting of all 
rules of the form 
\begin{equation}
H_i^{aw} \rightarrow \prod_{j=1}^J H_{i_j}^{w_j}
\label{e-generic_rule}
\end{equation}
which belong to the set of equations (\ref{compositionnal_recurrence21}) and $\bigcup_{j=1}^J \Al(i_j,w_j) \neq \emptyset$.\\
Let us note by ${\cal C}_2$ the rewriting system consisting of all 
rules of the form (\ref{e-generic_rule})
which belong to the set of equations (\ref{compositionnal_recurrence21}) and $\bigcup_{j=1}^J \Al(i_j,w_j) =\emptyset$.\\
Similarly ${\cal C}_3$ corresponds to the set of equations (\ref{compositionnal_recurrence22})
oriented from left to right and finally
$${\cal C} := {\cal C}_1 \cup {\cal C}_2 \cup {\cal C}_3.$$

This subsection proves termination of the system ${\cal C}$. Our proof-strategy consists in  relating derivations (modulo $\rightarrow_\mathcal{C}$)
with derivations (modulo $\rightarrow_\mathcal{A}$). We shall essentially show that, given a starting term $H_i^w$,
and some starting derivation $D$ of the form 
$\WD \rightarrow_\mathcal{A}^* H_i^w(\WD)$,
every derivation $H_i^w\rightarrow^n_\mathcal{C}\prod_{\lambda=1}^\Lambda H_{i_\lambda}^{w_\lambda}$ 
leads to a non-trivial {\em refinement} of size $J \geq n$ for the fixed derivation $D$.
This shows that $n$ cannot exceed the length of $D$.
We introduce below a notion of {\em derivation-path} and notions of {\em types} for derivations and derivation-paths,
which will allow a smooth inductive proof of the above refinement property (see Lemma \ref{l-fromC_toA}).

\begin{figure}
\begin{center}
\begin{tikzpicture}[scale=0.40]
\node[left] at (0,12){$H_i^w$};
\draw [->] (0,12) -- (8,24);
\draw [->] (0,12) -- (8,2);
\node at (4,12){$\longrightarrow^n_{{\cal C}}$};
\node at (10,2){$H_{i_J}^{w_J}$};
\node at (10,10){$H_{i_{j+1}}^{w_{j+1}}$};
\node at (10,14){$H_{i_{j}}^{w_{j}}$};
\node at (10,20){$H_{i_{2}}^{w_{2}}$};
\node at (10,24){$H_{i_{1}}^{w_{1}}$};
\node at (10,24){$H_{i_{1}}^{w_{1}}$};
\node at (10,4){$\circ$};
\node at (10,6){$\vdots$};
\node at (10,8){$\circ$};
\node at (10,12){$\circ$};
\node at (10,16){$\vdots$};
\node at (10,18){$\circ$};
\node at (10,22){$\circ$};
\node at (20,2){$D_J\;\;\;\;{\tiny +}\downarrow{\cal A}$};
\node at (20,6){$\vdots$};
\node at (20,10){$D_{j+1}\;\;\;\;{\tiny +}\downarrow{\cal A}$};
\node at (20,14){$D_{j}\;\;\;\;{\tiny +}\downarrow{\cal A}$};
\node at (20,16){$\vdots$};
\node at (20,20){$D_2\;\;\;\;{\tiny +}\downarrow{\cal A}$};
\node at (20,24){$D_1\;\;\;\;{\tiny +}\downarrow{\cal A}$};
\draw [-] (15,0) -- (25,0);
\draw [-] (15,4) -- (25,4);
\draw [-] (15,8) -- (25,8);
\draw [-] (15,12) -- (25,12);
\draw [-] (15,18) -- (25,18);
\draw [-] (15,22) -- (25,22);
\draw [-] (15,25) -- (25,25);
\node[right] at (26,14) {$\varepsilon_j(T_j,w_j,W_j) + \bar{\varepsilon}_j\bar{\tau}$};
\node[right] at (25,10) {$\varepsilon_{j+1}(T_{j+1},w_{j+1},W_{j+1}) + \bar{\varepsilon}_{j+1}\bar{\tau}$};
\node at (30,6){$\vdots$};
\node[right] at (26,2) {$\varepsilon_J(T_J,w_J,W_J) + \bar{\varepsilon}_J\bar{\tau}$};
\node at (30,16){$\vdots$};
\node[right] at (26,20) {$\varepsilon_2(T_2,w_2,W_2) + \bar{\varepsilon}_2\bar{\tau}$};
\node[right] at (26,24) {$\varepsilon_1(T_1,w_1,W_1) + \bar{\varepsilon}_1\bar{\tau}$};
\node at (0,26){Derivation (mod ${\cal C}$)};
\node at (10,26){Seq of morphisms};
\node at (20,26){Deriv-path (mod ${\cal A}$)};
\node at (30,26){Type};
\end{tikzpicture}
\caption{Derivations (modulo  ${\cal C}$) versus derivation-paths (modulo ${\cal A}$).}
\end{center}
\label{fig-fromC-toA}
\end{figure}
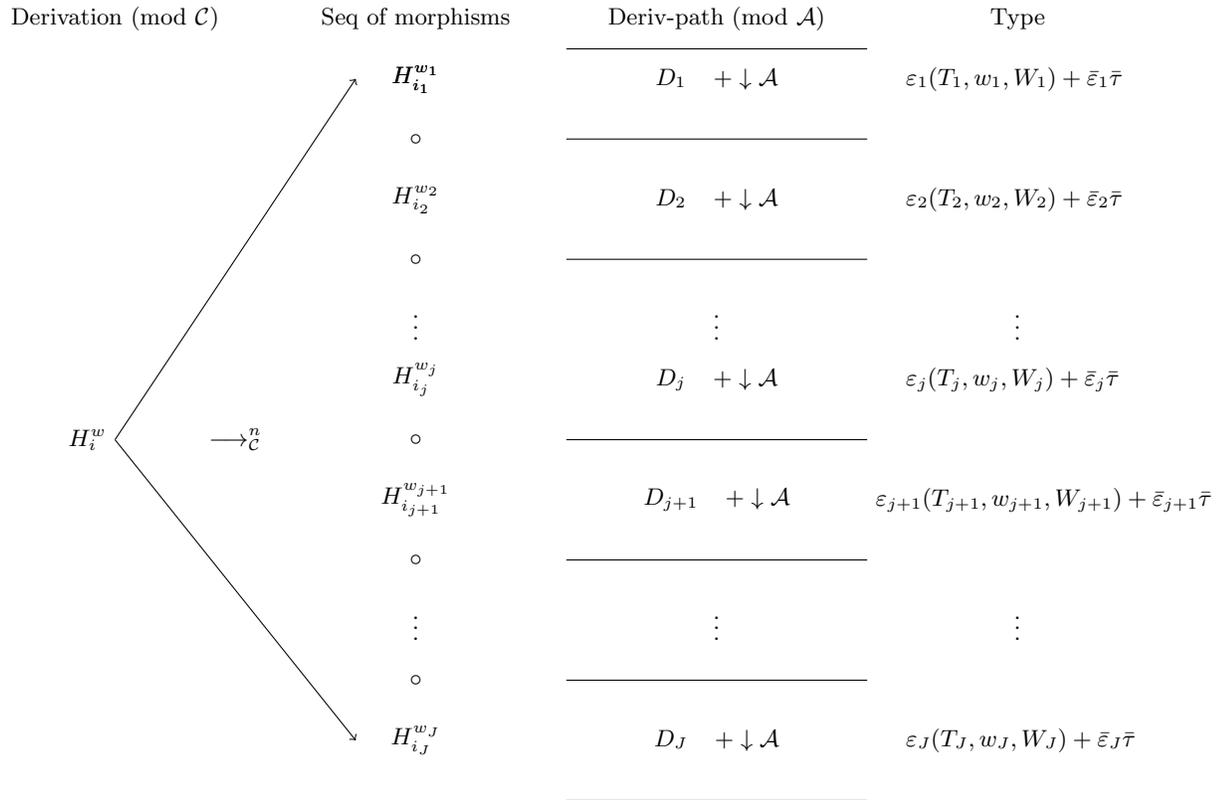
Figure \ref{fig-fromC-toA} depicts the relationship between some ${\cal C}$-derivation $\longrightarrow^n_{{\cal C}}$
starting from $H_i^w$ and some refinement of length $J \geq n$ of the derivation $\WD \rightarrow_{{\cal A}}^*H_i^w(\WD)$
(for some fixed well-choosen variable-term-word of order 3 $\WD$); the numbers $\varepsilon_j,\bar{\varepsilon}_j$ take their value
in $\{0,1\}$ and $\sum_{j=1}^J\bar{\varepsilon}_j\geq n_1$ (where $n_1$ is the number of rules from ${\cal C}_1$ used in the ${\cal C}$-derivation).\\

\paragraph{Types and undeterminates}
We define a finite set of types ${\cal T}$ by
$${\cal T} := \{ (T,w,W) \mid T \in \Gamma_2,w \in \Gamma_3^*,W \in {\cal W}\}\cup 
\{\varepsilon, \bar{\tau}\}$$
where $\bar{\tau}$ is just a new symbol. 

We manipulate here
terms (and variable-terms) of various levels in $[1,3]$, with pushdown symbols in $\Gamma$ and undeterminates in $E$.
Since $E \subseteq \hat{\Gamma}^*/\equiv_2$ the self-extension $\hat{\equiv_2}$ is defined over $\hat{\Gamma}^**E^*$,
which includes all the sets $\ell-\term(\Gamma,E)$ for $1 \leq \ell \leq 3$.
In the following we keep the same notation $\equiv_2$ for what should be denoted by $\hat{\equiv_2}$
(see Definition (\ref{eq-extension-congruence}) in \S \ref{par-prelimin-congruences}).
\begin{definition}
Let $D$ be a derivation (modulo $\rightarrow_\mathcal{A}$) over $(3-\vterm(\Gamma,E))^*$. Its type $\tau(D)$ is the following 
element of ${\cal T}$:\\
$\tau(D):= (T,w,W)$ if $D$ has the form $L\cdot W[T[w]u/e] \cdot R \rightarrow^*_\mathcal{A} 
L \cdot H_i^w(W)[u/e] \cdot R$
where $W=pS[e]q$,$W \in \Al(i,w)$, $u \in 2-\term(\Gamma , E), u \in e$ and $L,R \in (3-\vterm(\Gamma, E))^*$\\
$\tau(D):= \bar{\tau}$ if $D$ has not the above form and $\ell(D) \geq 1$\\
$\tau(D):= \varepsilon$ if $\ell(D)=0.$
\label{def-type-of-derivation}
\end{definition}
\begin{lemma}
The binary relation $\tau$ introduced by Definition \ref{def-type-of-derivation} is a map:
every derivation has exactly one type.
\label{lem-just-one-type}
\end{lemma}
\begin{proof}
Suppose that $A,B \in  (3-\vterm(\Gamma,E))^*$ are the extreme points of a derivation $D: A \rightarrow^*_\mathcal{A} B$
that admits both types $(T,w,W)$ and $(T',w',W')$. Let us show that these types are equal.
This means
$$A= L\cdot W[T[w]u/e] \cdot R \rightarrow^*_\mathcal{A} 
L \cdot H_i^w(W)[u/e] \cdot R=B\;\;$$
and
$$ A= L'\cdot W'[T'[w']u'/e'] \cdot R' \rightarrow^*_\mathcal{A} 
  L' \cdot H_{i'}^w(W')[u'/e'] \cdot R=B .$$
Let us assume that $W=pS[e]q, W'=p'S'[e']q'$\\
{\bf Case 1}: $|L| = |L'|$.\\
The two lefthandsides of rules of $\rightarrow_{\mathcal A}$ occur at the same position of $A$,
hence $W[T[w]u/e]=W'[T'[w']u'/e'] $. It follows that $pS[T[w]u]q=p'S'[T'[w']u']q'$
hence that $p=p', q=q', S=S', T=T', w=w'$. Since $e=[T[w]u]_{\equiv_2}$ and $e'=[T'[w']u']_{\equiv_2}$,
we also have $e=e'$, implying that $W=W'$.\\
{\bf Case 2}: $|L|< |L'|$.\\
By equidivisibility of the monoid $(3-\vterm(\Gamma,E))^*$, there exists some $M \in (3-\vterm(\Gamma,E))^*$, such that:
$$A=L\cdot  W[T[w]u/e] \cdot M \cdot W'[T'[w']u'/e'] \cdot R'$$
and
$$L'= L\cdot  W[T[w]u/e] \cdot M,\;\;R= M \cdot W'[T'[w']u'/e'] \cdot R'.$$
Let us consider the derivation (\ref{derivational_condition}) used for the definition of $H_i^w(pS[e]q)$:
$$(p S[T[w] \Omega]q) \rightarrow_{\mathcal{A}}^+ \prod_{j=1}^{\ell(i,W)} (p_{i,j}S_{i,j}[\Omega] q_{i,j}).$$
The word $B$ has thus a prefix of the form:
$$P:=L\cdot p_{i,1}S_{i,1}[u] q_{i,1}.$$
Since $B$ is also obtained from $A$ by rewriting the suffix $M \cdot W'[T'[w']u'/e'] \cdot R'$, and this derivation
leaves invariant the prefix of length $|L|+1$ of $A$, $B$ has a prefix of the form
$$P'= L\cdot pS[T[w]u]q.$$
Note that the atomic pushdowns  of order 3 $S_{i,1}[u],S[T[w]u]$ are different (because their level 2 are different).
The two prefixes $P,P'$  have same length and have a different last letter : this is impossible.\\
By symmetry $|L|> |L'|$ is also impossible.\\
Clearly the two other possible types ($\bar{\tau}$ and $\varepsilon$ ) are incompatible with a type $(T,w,W)$,
and also mutually incompatible. Hence, every derivation has at most one type.\\
It is clear from the definition that it has at least one type.
\end{proof}
\paragraph{Derivations and derivation-paths}
The {\em partial product} of two derivations is defined as usual and denoted by $\otimes$.
We denote by a dot-symbol $\cdot$ the left-operation of $(3-\vterm(\Gamma,E))^*$ on derivations: it consists in concatenating the left operand to the left, at every step of the derivation. 
The right-operation is also denoted by $\cdot$.
We call {\em derivation path} every sequence $(D_1,\ldots,D_n,\ldots,D_\ell)$ of derivations
such that the product $D_1\otimes  \cdots \otimes D_n \otimes \cdots \otimes D_\ell$ is defined 
(equivalently, such that the product $D_i \otimes D_{i+1}$ is defined for 
every $i \in [1,\ell-1]$).
The {\em partial product} of  two derivation-paths, denoted by $\odot$, is just the concatenation
of the two derivation-paths, when this concatenation is itself a derivation-path, and it is  otherwise undefined.\\
We define the ordering $\preceq$ over derivation-paths as the least ordering, which is compatible 
with the (partial) product $\odot$ and which fulfils that, for every pair of derivations $D,D'$, if  the product
$D \otimes D'$ is defined, then 
$$(D \otimes D') \preceq (D,D').$$

The reader shoud notice the analogy between words and words of level 2 (on one hand), derivations and derivation-paths (on the other hand).
A derivation-path can be interpreted as some {\em decomposition}  of a derivation.
An inequality $\vec{D} \preceq \vec{D'}$ means that both derivation-paths $\vec{D}, \vec{D'}$ are decompositions of the {\em same} derivation
(modulo $\rightarrow_\mathcal{A}$) and the later {\em refines} the former.
\paragraph{Typing the derivation-paths}
We extend the map $\tau$ into a map from the set of derivation-paths to the free commutative monoid
generated by ${\cal T}$,
$\bbbN\langle {\cal T}\rangle$, in such a way that $\tau$ is an homomorphism (of partial monoids).
Namely: for every derivation-path $(D_1,\ldots,D_n,\ldots,D_\ell)$:
$$\tau( D_1,\ldots,D_n,\ldots,D_\ell):= 
\sum_{n=1}^\ell \tau(D_n).$$
Let us use the notation:\\
$$t(i,w):= \sum_{W \in \Al(i,w)} (T,w,W)$$ for every $i\in I_0$ of the form 
$({\cal V},T)$ and every $w \in \Gamma^*$,
$$t(i,w)=0$$
 for every $i\in I \setminus I_0$ and every $w \in \Gamma^*$.
We also note $\al(i,w) := \card(\Al(i,w))$.
We denote by $\leq$ the product ordering over $\bbbN \langle {\cal T}\rangle$:
$$(\sum_{t\in {\cal T}} n_t \cdot t \leq \sum_{t\in {\cal T}} m_t \cdot t)
\Leftrightarrow
(\forall t \in {\cal T}, n_t \leq m_t).
$$
Let us recall that, by Lemma \ref{substitution_principle}, if
$D$ is a derivation (modulo  $\rightarrow^*_\mathcal{A}$), $e\in E$ and $u \in 2-\pds(\Gamma,E)$, 
then $D[u/e]$ is also a derivation (modulo $\rightarrow^*_\mathcal{A}$).
\begin{lemma}
\label{lem-type-is-morphic}
Let us consider a derivation $D$ (modulo  $\rightarrow^*_\mathcal{A}$), some variable-term-words of order 3 $L,R$ (over ($\Gamma,E$),
an undeterminate $e \in E$, a word $u \in \Gamma_3^*$ and two derivation-paths $\vec{D},\vec{D}'$. Then\\
1- $\tau(D[u/e]) = \tau(D)$\\
2- $\tau(L \cdot D \cdot R) = \tau(D)$\\
3- $\tau( \vec{D} \odot \vec{D}')=\tau( \vec{D}) + \tau(\vec{D}')$. 
\end{lemma}
This lemma is an immediate consequence of the definitions of $\tau$.
\begin{lemma}
Let $H_i^{aw}(W) = \prod_{j=1}^{\ell(i,a,d,W)} 
H_{\alpha(i,a,d,W,j)}^{\omega_{i,a,d,W,j} \cdot w} \circ \Phi_{i,a,d,W,j} 
\circ H_{\beta(i,a,d,W,j)}^{\omega_{i,a,d,W,j} \cdot w}(V_{i,a,d,W,j})$
be a ``mixed rule'' with $i=({\cal V},T),W \in \Al(i,w)$.
Let $D$ be a derivation (modulo $\rightarrow_\mathcal{A}$) of type $(T,aw,W)$.\\
If the rule realizes Case 1.1, then
there exists a derivation-path $\vec{D}' \succeq (D)$ such that:
$$\tau(\vec{D}') \geq \sum_{j=1}^{\ell(i,a,d,W)} t(\alpha(i,a,d,W,j),\omega_{i,a,d,W,j} \cdot w)
+ t(\beta(i,a,d,W,j),\omega_{i,a,d,W,j} \cdot w).$$
If the rule realizes Case 1.2 or 1.3 or 2,
then
there exists a derivation-path $\vec{D}' \succeq (D)$ such that:
$$\tau(\vec{D}') \geq (\sum_{j=1}^{\ell(i,a,d,W)} t(\alpha(i,a,d,W,j),\omega_{i,a,d,W,j} \cdot w)
+ t(\beta(i,a,d,W,j),\omega_{i,a,d,W,j} \cdot w))+ \bar{\tau}.$$
\label{l-mixedrule_increases_the_type}
\end{lemma}
\begin{proof}
Let us consider a rule and a derivation $D$ fulfiling the hypothesis of the lemma.
By Lemma \ref{l-mixed_recurrence_exists} we know that one of cases 1.1,1.2,1.3,2 must occur.\\
{\bf Case 1.1}:\\
Then $\sum_{j=1}^{\ell(i,a,d,W)} t(\alpha(i,a,d,W,j),\omega_{i,a,d,W,j} \cdot w)
+ t(\beta(i,a,d,W,j),\omega_{i,a,d,W,j} \cdot w)=0$.
Taking $\vec{D}' := (D)$, the conclusion of the lemma is true.\\
{\bf Case 1.2}:\\
Let
$D_0: W[T[aw]e/e] \rightarrow_\mathcal{A} V_{i,a,d,W,1}[T'[\omega_{i,a,d,W,j}w]e/e] \rightarrow^*_\mathcal{A} H_i^{aw}(W)[e/e].$
Let us decompose $D$ as $D_1 \otimes D_2$ with:
$$D_1:W[T[aw]e/e]\rightarrow_\mathcal{A} V_{i,a,d,W,1}[T'[\omega_{i,a,d,W,j}w]e/e] $$
$$D_2: V_{i,a,d,W,1}[T'[\omega_{i,a,d,W,1}w]e/e] \rightarrow^*_\mathcal{A} H_i^{aw}(W)[e/e]
= H_{\alpha(i,a,d,W,1)}^{\omega_{i,a,d,W,1} \cdot w} (V_{i,a,d,W,1})[e/e].$$
Since the extremity of $D_1$ is not of the form $H_i^{aw}(W)[e/e] $,
$\tau(D_1)=\bar{\tau}$.
We also see that: $\tau(D_2)=(T',\omega_{i,a,d,W,j}w,V_{i,a,d,W,1})$.
Let $\vec{D}':= (D_1,D_2)$.
This derivation-path fulfils:
$$ \vec{D}' \succeq (D)$$
and
\begin{eqnarray*}
  \tau(\vec{D}')& = & \tau(D_1) + \tau(D_2) \\
  & =& \bar{\tau} + (T',\omega_{i,a,d,W,j}w,V_{i,a,d,W,1})\\
  & =& (T',\omega_{i,a,d,W,j}w,V_{i,a,d,W,1}) + \bar{\tau}.
\end{eqnarray*}
But the value of $t(\beta(i,a,d,W,1),\omega_{i,a,d,W,1} \cdot w)$ is $0$,
hence $\vec{D}'$ fulfils the conclusion of the lemma.\\
A general derivation $D$ of type $(T,aw,W)$, fulfilling this Case, must be of the form
$D = L \cdot D_0[u/e] \cdot R\;\;$ for some $L,R \in (3-\vterm(\Gamma, E))^*, u \in 2-\term(\Gamma, E)$
and $ u \equiv_{2}e$.
Since the substitution $[u/e]$ and the product by $L$, on the left, and by $R$, on the right,
preserve the operation $\otimes$ and preserve the type, we obtain the required conclusion for $D$.\\
{\bf Case 1.3}:\\
Let  
$D_0: W[T[aw]e/e] \rightarrow_{\mathcal{A}}
V[T'[w']T''[w']e/e'] \rightarrow^*_\mathcal{A} H_i^{aw}(W),$\\
where, for ease of notation, we have used the abbreviation:
$w' := \omega_{i,a,d,W,j}w,\;V:= V_{i,a,d,W,1}$.\\
Let us also abbreviate 
$\alpha := \alpha(i,a,d,W,1);\;\Phi:= \Phi_{i,a,d,W,1}\;\;\beta:= \beta(i,a,d,W,1)$.
We define two derivations $C_1,C_2$ by:
$$
C_1: W[T[aw]e/e] \rightarrow_{\mathcal{A}}
V[T'[w']T''[w']e/e']$$
$$C_2: V[T'[w']T''[w']e/e'] \rightarrow^*_\mathcal{A} H_{\alpha}^{w'}(V) [T''[w']e/e']$$
The word $H_{\alpha}^{w'}(V)$ decomposes into letters as:
$$H_{\alpha}^{w'}(V) := W_1 \cdots W_\lambda \cdots W_\Lambda$$ where each $W_\lambda$ has the form
$W_\lambda=p_\lambda S_\lambda[e']q_\lambda$.
Let us define
$$W_\lambda':= W_\lambda[e/e'].$$
We then have the derivations $D_\lambda$ defined by:
\begin{equation}
D_\lambda: \;\;\;W_\lambda[T''[w']e/e']= W'_\lambda[T''[w']e/e]\rightarrow^*_\mathcal{A} H_\beta^{w'}(W_\lambda')
\label{e-derivation_Dlambda}
\end{equation} 
By definition of Case 1.3 in Lemma \ref{l-mixed_recurrence_exists}, $\Al(\beta)=\{\Phi(W_\lambda) \mid 1 \leq \lambda \leq \Lambda\}$
hence $\Al(\beta)=\{W'_\lambda \mid 1 \leq \lambda \leq \Lambda\}$.
Let us denote by $D_\lambda$ the derivation (\ref{e-derivation_Dlambda}). We then have
\begin{equation}
D_0= C_1 \otimes C_2 \otimes \bigotimes_{\lambda=1}^\Lambda D_\lambda
\label{e-D0_decomposition1.3}
\end{equation}
and
\begin{equation}
\tau(C_1)=\bar{\tau},\;\tau(C_2)=(T',w',V),\;\tau(D_\lambda)=(T'',w',W'_\lambda) 
\end{equation}
Let us define 
$$\vec{D}' := (C_1,C_2,D_1,\ldots,D_\lambda,\ldots,D_\Lambda)$$
By (\ref{e-D0_decomposition1.3}), $\vec{D}' \succeq (D_0)$.
Since $\Al(\alpha,w')=\{V\}$ and $\Al(\beta,w')=\{W'_\lambda \mid 1 \leq \lambda \leq \Lambda\}$,
we can write: 
\begin{eqnarray}
  \tau(\vec{D}') & = & \tau(C_1) + \tau(C_2) + \sum_{\lambda=1}^\Lambda \tau(D_\lambda)\nonumber\\
                 & = & \bar{\tau} + t(\alpha,w') + t(\beta,w')\nonumber\\
                & \geq & t(\alpha,w')  + t(\beta,w') + \bar{\tau} 
\end{eqnarray}
as required.
We reduce to $D_0$ the  treatment of a general derivation $D$ of type $(T,aw,W)$, fulfilling Case 1.3, as we did for Case 1.2.\\ 
{\bf Case 2}:\\
Let  
$D_0: W[T[aw]e/e] \rightarrow_{\mathcal{A}}
V_1[T[w']e/e] V_2[T[w']e/e] \rightarrow^*_\mathcal{A} H_i^{aw}(W)$
where, for ease of notation, we have used the abbreviation:
$w' := \omega_{i,a,d,W,j}\cdot w,\;V_1:= V_{i,a,d,W,1},\;V_2:= V_{i,a,d,W,2}$.\\
Let us also abbreviate 
$\alpha_1 := \alpha(i,a,d,W,1);\;\alpha_2 := \alpha(i,a,d,W,2),
\;\;\beta_1:= \beta(i,a,d,W,1),\;\;\beta_2:= \beta(i,a,d,W,2)$.
We define three derivations $C_1,D_1,D_2$ by:
$$C_1: W[T[aw]e/e] \rightarrow_{\mathcal{A}}^+
V_1[T[w']e/e] V_2[T[w']e/e]
$$
$$
D_1:  V_1[T[w']e/e] \rightarrow^* _{\mathcal{A}} H_{\alpha_1}^{w'}(V_1).$$
$$
D_2:  V_2[T[w']e/e] \rightarrow^* _{\mathcal{A}} H_{\alpha_2}^{w'}(V_2).$$
We then have
\begin{equation}
D_0= C_1 \otimes D_1 \otimes D_2
\label{e-D0_decomposition2}
\end{equation}
and
\begin{equation}
\tau(C_1)=\bar{\tau},\;\tau(D_1)=(T,w',V_1),\;\tau(D_2)=(T,w',V_2).
\end{equation}
Let us define 
$$\vec{D}' := (C_1,D_1,D_2)$$
Since $\Al(\alpha_1,w')=\{V_1\}$,$\Al(\alpha_2,w')=\{V_2\}$  and $\Al(\beta_1,w')= \Al(\beta_2,w')=\emptyset$,
we can write: 
\begin{eqnarray}
\tau(\vec{D}') & = & \tau(C_1) + \tau(D_1) + \tau(D_2) \nonumber\\
          & = & \bar{\tau} + t(\alpha_1,w')  + t(\alpha_2,w') \nonumber\\
          & = & \biggl( \sum_{j=1}^2 (t(\alpha_j,w')+ t(\beta_j,w'))\biggr) +\bar{\tau}
\end{eqnarray}
as required.
We reduce to $D_0$ the  treatment of a general derivation $D$ of type $(T,aw,W)$, fulfilling Case 2, as we did for Case 1.2.\\ 
\end{proof}
We treat in next lemma the case of a rule of ${\cal C}_1$.
\begin{lemma}
Let us consider the compositional rule 
$H_i^{aw} \rightarrow \prod_{j=1}^J H_{i_j}^{w_j}$ , given in (\ref{compositionnal_recurrence21}), 
corresponding to a class $d$ and an index $i \in I_0$
of the form $i=({\cal V},T)$ and suppose that $\bigcup_{j=1}^J \Al(i_j,w_j) \neq \emptyset$.
Let $\vec{D}$ be a derivation-path (modulo $\rightarrow_\mathcal{A}$) and let $\tau_0 \in 
\bbbN\langle {\cal T} \rangle$ such that
$$\tau(\vec{D})= t(i,aw)+ \tau_0.$$
Then, there exists a derivation-path $\vec{D}' \succeq \vec{D}$ such that
$$\tau(\vec{D}') \geq \sum_{j=1}^{J} t(i_j,w_j)+ \bar{\tau}+\tau_0.$$
\label{l-compositionalrule1_increases_the_type}
\end{lemma}
\begin{proof}
Let $\vec{D}$ fulfil the hypothesis of the lemma. It has the form
$$\vec{D}=\vec{E}_0 \odot \biggl(\bigodot_{\lambda=1}^{\al(i,w)} (D_\lambda) \odot \vec{E}_\lambda \biggr)$$
with $\Al(i,w)=\{W_\lambda \mid 1 \leq \lambda \leq \al(i,w)\}$,
$$\sum_{\lambda=0}^{\al(i,w)} \tau(E_\lambda)=\tau_0;\;\;
\tau(D_\lambda)= (T,aw,W_\lambda) \mbox{ for } 1 \leq \lambda \leq \al(i,w).$$
By Lemma \ref{l-mixedrule_increases_the_type}, for every $\lambda \in [1,\al(i,w)]$, 
there exists a derivation-path $\vec{D}'_\lambda \succeq (D_\lambda)$ such that:
\begin{equation}
\tau(\vec{D}'_\lambda) \geq \sum_{j=1}^{\ell(i,a,d,W_\lambda)} t(\alpha(i,a,d,W_\lambda,j),\omega_{i,a,d,W_\lambda,j} \cdot w)
+ t(\beta(i,a,d,W_\lambda,j),\omega_{i,a,d,W_\lambda,j} \cdot w).
\label{e-local_lb}
\end{equation}
Since the righthanside of the given rule 
$$H_i^{aw} \rightarrow \prod_{j=1}^J H_{i_j}^{w_j}$$
contains, as disjoint factors, all the products 
$H_{\alpha(i,a,d,W_\lambda,j)}^{\omega_{i,a,d,W_\lambda,j} \cdot w} 
\circ H^w_{i,a,d,W,j,0}
\circ H_{\beta(i,a,d,W_\lambda,j)}^{\omega_{i,a,d,W_\lambda,j}}$,
and no other factor with index $i'$ and exponent $w'$ such that $t(i',w') \neq 0$,
we get that
\begin{equation}
\sum_{j=1}^{J} t(i_j,w_j)= \sum _{\lambda=1}^{\al(i,w)} (\sum_{j=1}^{\ell(i,a,d,W_\lambda)} t(\alpha(i,a,d,W_\lambda,j),\omega_{i,a,d,W_\lambda,j} \cdot w)
+ t(\beta(i,a,d,W_\lambda,j),\omega_{i,a,d,W_\lambda,j} \cdot w)).
\label{e-rewritetypes_vs_rewritehoms}
\end{equation}
Let us define
$$\vec{D}':=\vec{E}_0 \odot \biggl(\bigodot_{\lambda=1}^{\al(i,w)} \vec{D}'_\lambda \odot \vec{E}_\lambda \biggr).$$
This derivation-path fulfils $\vec{D}' \succeq \vec{D}$ and, by (\ref{e-local_lb},\ref{e-rewritetypes_vs_rewritehoms})
\begin{equation}
\tau(\vec{D}') \geq \biggl( \sum_{j=1}^{J} t(i_j,w_j) \biggr) + \tau_0.
\label{e-global_lb1}
\end{equation}
Assuming that $\bigcup_{j=1}^J \Al(i_j,w_j) \neq \emptyset$, at least one index $j$ cooresponds to some case other than case 1.1 in
Lemma \ref{l-compositionalrule1_increases_the_type}. Hence,  at least one of the minorations
(\ref{e-local_lb}) can be improved by adding $\bar{\tau}$ to its righhand-side.
Hence
$$
\tau(\vec{D}') \geq \biggl(  \sum_{j=1}^{J} t(i_j,w_j) \biggr) + \tau_0 + \bar{\tau}.
\label{e-global_lb2}
$$
\end{proof}
Let us treat now the case of a rule of ${\cal C}_2 \cup{\cal C}_3$.
\begin{lemma}
Let us consider the compositional rule 
$H_i^{aw} \rightarrow \prod_{j=1}^J H_{i_j}^{w_j}$, 
either given in (\ref{compositionnal_recurrence21}), 
by a class $d$ and an index $i \in I_0$
of the form $i=({\cal V},T)$ such that $\bigcup_{j=1}^J \Al(i_j,w_j)=\emptyset$
or given in (\ref{compositionnal_recurrence22}).
Let $\vec{D}$ be a derivation-path (modulo $\rightarrow_\mathcal{A}$) and some $\tau_0 \in 
\bbbN\langle {\cal T} \rangle$ such that
$$\tau(\vec{D})= t(i,aw)+ \tau_0.$$
Then, there exists a derivation-path $\vec{D}' \succeq \vec{D}$ such that
$$\tau(\vec{D}') \geq \sum_{j=1}^{J} t(i_j,w_j)+ \tau_0.$$
\label{l-compositionalrule2_increases_the_type}
\end{lemma}
\begin{proof}
Let us choose $\vec{D}':= \vec{D}$. Since $\bigcup_{j=1}^J \Al(i_j,w_j)=\emptyset$,
in fact $\sum_{j=1}^{J} t(i_j,w_j)=0$ and the lemma follows.
\end{proof}
\begin{lemma}
Suppose that $H_i^w \rightarrow^*_\mathcal{C}\prod_{\lambda=1}^\Lambda H_{i_\lambda}^{w_\lambda}$ 
and that this derivation (modulo $\rightarrow_\mathcal{C}$) has $n$ steps in 
$\rightarrow_{\mathcal{C}_1}$.
Let $\vec{D}$ be a derivation-path (modulo $\rightarrow_\mathcal{A}$) such that
$$\tau(\vec{D})= t(i,w).$$
Then, there exists a derivation-path $\vec{D}' \succeq \vec{D}$ such that
$$\tau(\vec{D}') \geq \sum_{\lambda=1}^{\Lambda} t(i_\lambda,w_\lambda)+ n \cdot \bar{\tau}.$$
\label{l-fromC_toA}
\end{lemma}
\begin{proof}
We prove this lemma by induction on the length $m$ of the given derivation ( modulo $\rightarrow_{\cal C}$).
If we suppose that $m=0$, then $\Lambda=1,i_\lambda=i, w_\lambda=w,n=0$, so that the conclusion of the lemma does hold.\\
Let us consider a derivation of length $m+1$ with $n$ steps in $\rightarrow_{{\cal C}_1}$.
 It can be decomposed as
$$H_i^w \rightarrow^*_\mathcal{C}\prod_{\lambda=1}^\Lambda H_{i_\lambda}^{w_\lambda}
\rightarrow_\mathcal{C}\prod_{\lambda=1}^{\ell-1} H_{i_\lambda}^{w_\lambda}\circ
\prod_{j=1}^J H_{\iota_j}^{v_j} \circ \prod_{\lambda=\ell+1}^\Lambda H_{i_\lambda}^{w_\lambda}$$
where
\begin{equation}
H_{i_{\ell}}^{w_{\ell}} \rightarrow_\mathcal{C}
\prod_{j=1}^J H_{\iota_j}^{v_j}.
\label{e-lemma10_last_step}
\end{equation} 
Let $\vec{D}$ fulfil the hypothesis of the lemma.\\ 
{\bf Case 1}: The last step (\ref{e-lemma10_last_step}) uses $\rightarrow_{{\cal C}_1}$.\\
By induction hypothesis,
there exists a derivation-path $\vec{D}'' \succeq \vec{D}$ such that
$$\tau(\vec{D}'') \geq \sum_{\lambda=1}^{\Lambda} t(i_\lambda,w_\lambda)+ (n-1) \cdot \bar{\tau}.$$
Using  Lemma \ref{l-compositionalrule1_increases_the_type} with 
$\tau_0=\tau(\vec{D}'')-t(i_\ell,w_\ell)$, we obtain a derivation-path $\vec{D}' \succeq \vec{D}''$ such that
\begin{eqnarray*}
\tau(\vec{D}') & \geq & \sum_{j=1}^J t(\iota_j,v_j)+ \tau_0 +\bar{\tau}\\
& = &
\sum_{\lambda=1}^{\ell-1} t(i_\lambda,w_\lambda) +
\sum_{j=1}^J t(\iota_j,v_j) + \sum_{\lambda=\ell+1}^\Lambda t(i_\lambda,w_\lambda)+ n \cdot \bar{\tau}\end{eqnarray*}
as required.\\
{\bf Case 2}: The last step (\ref{e-lemma10_last_step}) uses $\rightarrow_{{\cal C}_2 \cup {\cal C}_3 }$.\\
By induction hypothesis,
there exists a derivation-path $\vec{D}'' \succeq \vec{D}$ such that
$$\tau(\vec{D}'') \geq \sum_{\lambda=1}^{\Lambda} t(i_\lambda,w_\lambda)+  n \cdot \bar{\tau}.$$
Using Lemma \ref{l-compositionalrule2_increases_the_type} we can conclude as in Case 1.
\end{proof}
For every $ i \in I_0$ and $w \in \Gamma_3^*$, we choose 
an enumeration, without repetition,
of the set $\Al(i,w)$:
$$ \Al(i,w)= \{W_1, W_2, \cdots, W_{\al(i,w)}\},$$
and define the word
$$ \WD(i,w)= W_1 W_2 \cdots W_{\al(i,w)}$$ 
\begin{lemma}
The relation $\rightarrow_{{\cal C}_2 \cup {\cal C}_3}$ is noetherian.
\label{lem-C2C3_isnotherian}
\end{lemma}
This is straightforward.
\begin{lemma}
The system of all relations (\ref{compositionnal_recurrence21}-\ref{compositionnal_recurrence22})
is a {\em noetherian} system of regular recurrent relations of order 1 in $\HOM(\hat{{\cal W}}^*,
\hat{{\cal W}}^*)$. 
\label{compositionnal_sys2_isnotherian}
\end{lemma}
\begin{proof}
Let us suppose that
\begin{equation}
\rightarrow_{{\cal C}} \mbox{ admits an infinite derivation }
\label{e-hypo_absurdum}
\end{equation}
The system ${\cal C}$ is monadic, hence there exists such an infinite derivation starting on
a letter $H_i^w$ (for some $i \in I_0, w \in \Gamma_3^*$):
\begin{equation}
H_i^w \rightarrow_{{\cal C}}^\infty
\label{e-inty_derC}
\end{equation}
By Lemma \ref{lem-C2C3_isnotherian}, this derivation (\ref{e-inty_derC}) must use infinitely many steps in $\rightarrow_{{\cal C}_1}$.
Let us consider a derivation 
$$D_i: \WD(i,w)[T[w] e/e; e \in 2-\pds(\Gamma)/\equiv_{2}] \rightarrow_{\mathcal{A}}^* H_i^w(\WD(i,w))$$
Since $\tau(D_i)=t(i,w)$, 
by Lemma \ref{l-fromC_toA}, for every $n \geq 0$, 
$$ \WD(i,w))[T[w] e/e; e \in 2-\pds(\Gamma)/\equiv_{2}]\rightarrow^{\geq n}_{\mathcal{A}} H_i^w(\WD(i,w)).$$
Since each step of a derivation modulo $\rightarrow_{\mathcal{A}}$ applies on only one letter,
this would imply that there exists a single letter $W_j$ such that, for every $n \geq 0$,
$$W_j [T[w] e_j /e_j] \rightarrow^{\geq n}_{\mathcal{A}}H_i^w(W_j).$$
Let $u _j \in 3-\pds(\Gamma)$ such that $u_j \in e_j$ and let us substitute $u_j$ to $e_j$ in the above derivation.
For every $n \geq 0$,
\begin{equation}
W_j [T[w] u_j /e_j] \rightarrow^{\geq n}_{\mathcal{A}}H_i^w(W_j)[u_j/e_j].
\label{e-Derivation_L11_1}
\end{equation}
By definition of $H_i^w(W_j)$, since $W_j \in \Al(i,w)$ ,
$$\Language(\mathcal{A}, H_i^w(W_j)[u_j/e_j]) \neq \emptyset.$$
Hence there exists a fixed word $v_j \in B^*$ and a fixed integer $m \geq 0$ such that
\begin{equation}
H_i^w(W_j)[u_j/e_j]\rightarrow^{m}_{\mathcal{A}}v_j.
\label{e-Derivation_L11_2}
\end{equation}
Finally, combining derivation (\ref{e-Derivation_L11_1}) with derivation (\ref{e-Derivation_L11_2}):
\begin{equation}
\forall n \geq m, W_j [T[w] u_j /e_j] \rightarrow^{\geq n}_{\mathcal{A}} v_j \in B^*.
\label{e-infinitelymany_terminal_derivations}
\end{equation}
Since $\mathcal{A}$ is strongly deterministic, there is only one maximal $\mathcal{A}$-computation
starting from $p_jS_j[T[w]u_j]$ and this computation is finite (it ends in the configuration with state $q_j$ and empty $\pds$). This entails that there are only finitely many derivations
modulo $\rightarrow_\mathcal{A}$ starting from  $p_jS_j[T[w]u_j]q_j$.
This contradicts assertion (\ref{e-infinitelymany_terminal_derivations}). We have thus proved that
hypothesis (\ref{e-hypo_absurdum}) is impossible.
\end{proof}
Let us complete the proof of $(1) \Rightarrow (2)$ in Theorem \ref{techcharacterisation_level3} by a suitable choice of
alphabets $A,C$, initial letter $c \in C$, main homomorphism sequence $w \mapsto H_{i_1}^w: A^*  \rightarrow \HOM(C^*,C^*)$ and final homomorphism $h\in \HOM(C^*,B^*)$ (while the alphabet $B$ is the terminal alphabet of ${\cal A}$).
We choose:
$$A:= \Gamma_3,\;\;C:= {\cal W}$$
$$e_0:= [\varepsilon]_{\equiv_2},\;\;c:=q_0\gamma_1[e_0]q_0,\;\;i_1:=({\cal W},\gamma_2)$$
and, forall $p,q \in Q, S \in \Gamma_1, e \in 2-\pds(\Gamma)/\equiv_{2}$
\begin{eqnarray}
h(pS[e]q)  & := & s \mbox{ if } (e=[\varepsilon]_{\equiv_2} \mbox{ and } pS[\varepsilon]q \rightarrow_{\mathcal{A}}^* s),\label{coacc-variables}\\
h(pS[e]q) & := & \varepsilon \mbox{ otherwise }
\label{noncoacc-variables}
\end{eqnarray}
We claim that, for every $w \in A^*$
\begin{equation*}
f(w) = h(H_{i_1}^w(c)).
\label{eq-f-from_Hiw}
\end{equation*}
(Note that the right-hand side of (\ref{noncoacc-variables}) can be choosen arbitrarily, without affecting this crucial claim).\\
\noindent Let us prove that $(2) \Rightarrow (3)\Rightarrow (4)\Rightarrow (5)$:\\
It suffices to use the implications $(2) \Rightarrow (3) \Rightarrow (4) \Rightarrow (5)$ 
from Theorem \ref{characterisation_level2}.\\

\noindent Let us sketch a proof that $(5)\Rightarrow (1)$:\\
Proposition 70 of \cite{Fra-Sen06} asserts that, if $f \in \mathbb{S}_k(A^* ,B^*),g  \in \mathbb{S}_\ell(B^* ,C^*)$ then 
$f\circ g \in \mathbb{S}_{k+\ell-1}(A^*,C^*)$, in the particular case where $|A| = |B| =|C|=1$ (i.e., for integer sequences).
By a suitable adaptation of the proof of this proposition, one can show that the same property holds for arbitrary finite alphabets $A,B,C$.
The case where $k=\ell=2$ gives the required implication $(5)\Rightarrow (1)$.

\section{Applications}
\begin{definition}[Polynomial recurrent relations]
Given a finite  index set $I=[1,n]$ and a family of mappings 
$f_i: A^* \rightarrow \bbbn$ (for $i \in I$),
we call {\em system of polynomial recurrent relations} a system of the form
$$f_i(aw)=P_{i,a}(f_1(w),f_2(w),\ldots,f_n(w)) \mbox{ for all } i\in I,a \in A, w \in A^*$$
where $P_i \in \bbbn[X_1,X_2,\ldots,X_n]$.
\label{polynomial_recurrence}
\end{definition}
A similar definition can be given for mappings $f_i: A^* \rightarrow \bbbz$ (for $i \in I$)
and polynomials $P_i \in \bbbz[X_1,X_2,\ldots,X_n]$.

Theorem \ref{characterisation_level3} specializes as follows in the particular case where $B$ is reduced to one letter i.e.
when the mapping $f$ is a formal power series.
\begin{corollary}
Let us consider a mapping $f:A^* \rightarrow \bbbn$.
The following properties are equivalent:\\
1- $f \in \mathbb{S}_3(A^*,\bbbn)$\\
2- There exists a finite family $(H_i)_{i\in [1,n]}$ of mappings $A^* \rightarrow \HOM(C^*,C^*)$ 
which fulfils a system of recurrent relations in $(\HOM(C^*,C^*),\circ,{\rm Id})$ ,
an element $h \in \HOM(C^*,\bbbn)$ and a letter $c \in C$ such that, for every $w \in A^*$:
$$f(w)=h(H_1(w)(c)).$$
3- $f$ is composition of a DT0L sequence $g: A^* \rightarrow C^*$ by a rational series $h: C^* \rightarrow \bbbn$.\\
4- There exists a finite family $(f_i)_{i\in [1,n]}$ of mappings $A^* \rightarrow \bbbn$ 
fulfilling a system of polynomial recurrent relations and such that $f=f_1$.
\label{characterisation_series_level3}
\end{corollary}
\begin{sketch}
{\bf $(1) \Rightarrow (2) \Rightarrow(3)$}:\\
this is exactly the sequence of implications $(1) \Rightarrow (2) \Rightarrow(3)$ of Theorem \ref{characterisation_level3},
reformulated in the case where $B$ is a unary alphabet.\\
{\bf $(3) \Rightarrow(4)$}:\\
Suppose that $f$ fulfills condition (3): $g=g_1$ for some family $(g_i)_{i \in I}$ of maps $A^* \rightarrow C^*$
which fulfills the system of catenative recurrent relations
\begin{equation*}
\forall i \in I, \forall a \in A,\forall w \in A^*,\;\;g_i(aw)= \prod_{j=1}^{\ell(i,a)} g_{\alpha(i,a,j)}(w) 
\label{eq-catenative-system}
\end{equation*}
while the $\bbbn$-rational series $h$ fulfills:
$$\forall w \in C^*,\;\;h(w)= L_0 \cdot H(w) \cdot C_0$$
for some $L_0 \in \bbbn^{1 \times d},H \in \HOM(C^*,\bbbn^{d \times d}), C_0 \in \bbbn^{d \times 1}$.\\
Let us denote by $u_{i,k,\ell}(w)$ the entry with indices $(k,\ell)$ of the $(d,d)$ matrix $H(g_i(w))$.
Since we have
\begin{equation*}
\forall i \in I,\forall a \in A,\forall w \in A^*,\;\;H(g_i(aw))= \prod_{j=1}^{\ell(i,a)} H(g_{\alpha(i,a,j)}(w))
\label{eq-matricial-system}
\end{equation*}
and the formula expressing a matricial product are polynomials (of degree 2), we get some polynomials $P_{i,k,\ell}$ with undeterminates $X_{i',j',k'}$
($i'\in I, j', k'\in [1,d]$) such that
\begin{equation}
\forall i \in I, \forall k,\ell \in [1,d],\forall a \in A,\forall i \in I,\;\;
u_{i,k,\ell}(aw))= P_{i,k,\ell}(\vec{u}(w))
\label{eq-polynomial-system}
\end{equation}
where $\vec{u}(w)$ is the $|I|\cdot d^2$-tuple of integers $(u_{*,*,*}(w))$.
Finally, if we note $M(u_i)$ the $(d,d)$ matrix with coefficients $u_{i,*,*}$, and we note $K$ the linear form
$\vec{u} \mapsto L_0 \cdot M(u_i)\cdot C_0$ we get that:
\begin{equation}
\forall i \in I, \forall w \in A^*,\;\;f_i(w)= K(\vec{u}(w)).
\label{eq-linearcomb-sequ}
\end{equation}
By (\ref{eq-polynomial-system}) the sequences $u_{i,k,\ell}$ are polynomially recurrent and by  (\ref{eq-linearcomb-sequ}), the sequences
$f_i$ are polynomially recurrent too.\\
{\bf $(4) \Rightarrow(1)$}:\\
Proposition 53 p. 384 of \cite{Fra-Sen06} asserts that, for  every family $(f_i)_{i\in [1,n]}$ of mappings $A^* \rightarrow \bbbn$ 
fulfilling a system of polynomial recurrent relations, $f_1$ belongs to $\mathbb{S}_3(A^*,\bbbn)$, in the particular case
of a {\em unary} alphabet $A$.
The extension of the proof, hence of the result, to an arbitrary finite alphabet $A$ is easy.\\
\end{sketch}
\begin{definition}
Let $\mathbb{S}$ be a set of mappings $A^* \rightarrow \bbbN$.
We denote by $\DIFF(\mathbb{S})$ the set of mappings
of the form:
$$ f(w) = g(w) - h(w)\;\;\mbox{ for all } w \in A^*,$$
for some mappings $g,h \in \mathbb{S}$.
We denote by $\FRAC(\mathbb{S})$ the set of mappings 
of the form:
$$ f(w)= \frac{g(w) - h(w)}{f'(w) - g'(w)} \;\;\mbox{ for all } w \in A^*,$$
for some mappings $f,g,f',g' \in \mathbb{S}$.
\end{definition}
Using point 4 of corollary \ref{characterisation_series_level3} we can prove the following
\begin{theorem}
The equality problem is decidable for formal power series in $\FRAC(\mathbb{S}_3(A^*,\bbbn))$.
\label{equivalence_level3}
\end{theorem}
The method consists, in a way similar to \cite{Sen99} or \cite{Hon00}, in reducing 
such an equality problem to deciding whether some polynomial belongs to the ideal
generated by a finite set of other polynomials.
Let us give a useful algebraic property of sequences of 
$\FRAC(\mathbb{S}_3(\bbbn,\bbbn))$.
\begin{theorem}[Th.11 of \cite{Cad-Maz-Pap-Pil-Sen20}]
Let $(u_n)_{n \in \bbbN}\in \FRAC(\mathbb{S}_3(\bbbN,\bbbN))$. Then there exists some $k \in \bbbN$ such that the sequences $u_n,u_{n+1},\ldots,u_{n+k-1}$ are algebraically dependent over $\bbbQ$.
\label{alg_dependency_level3}
\end{theorem}

\section{Examples and counter-examples}
\label{examples6}
We examine here four examples of mappings $A^* \rightarrow \bbbn$ and 
locate them in the classes $\mathbb{S}_k(A^*,\NN)$ or some related classes.\\
\begin{example}

The Fibonacci sequence $F_n$ defined by
$$F_0=1,F_1=1,F_{n+2}=F_{n+1}+F_n \mbox{ for all }n \geq 0
$$
is clearly in $\mathbb{S}_2$ since it fulfills a linear recurrence relation with 
coefficients in $\bbbn$ (i.e. a catenative recurrence where the alphabet has size $1$).
\end{example}
\begin{example}
Let $G:\{0,1\}^* \rightarrow \bbbn$ be defined by
$$G(w)=F_{\nu(w)}$$
where $\nu(w)$ is the natural number expressed by $w$ in base 2.
Since $\sum_{n=0}^\infty \nu(w) w$ is a rational series, $G$ fulfils point 3 of our characterisation of $\mathbb{S}_3(\{0,1\}^*,\bbbn)$.
\end{example}
\begin{example}
  \label{example:fact}
Factorial sequence:$FC(n) := (n+1)!$.\\
i.e. $FC(n)$ denotes ``factorial $(n+1)$''.\\
It has level $3$.\\
\end{example}
\begin{proof}
Let us define the auxiliary sequence $L(n) := n+2.$
These sequences fulfill the polynomial recurrence relations:
$$L(0)=2\;\;FC(0)= 1\;\;L(n+1) = L(n) +1\;\;FC(n+1)= L(n) \cdot FC(n).$$
Hence $FC$ is a sequence in $\mathbb{S}_3$.\
\end{proof}
By Theorem \ref{characterisation_level3}, $FC$ is the composition of two word-sequences of order $2$.\\
An explicit such decomposition is the following:\\
\begin{eqnarray*}
u(0)=b\;\;v(0)=ab\;\;& u(n+1)=u(n) \cdot v(n)&
\;\;v(n+1)=a \cdot v(n)\\
U(\varepsilon)=1\;\;V(\varepsilon)=0\;\;&&\\ 
U(a\cdot w)=U(w) + V(w)\;\;& V(a\cdot w)=V(w)\;\;&
U(b\cdot w)=U(w)\;\; V(b\cdot w)=V(w)\;\;
\end{eqnarray*}
We get:
$$u(n) = babaab \cdots ba^nb\;\;\;\;v(n) = a^{n+1}b$$
$$U(u(n))= (n+1)!$$
We thus have decomposed $FC$ into two recurrences of level 2 (going
``through'' words):
$$\forall n \in \NN,\;\;FC(n) = U(u(n)) .$$
Nevertheless, $FC$ is {\em not} the composition of two integer sequences of order $2$:
this can be deduced from the rate of growth of $FC$, which is strictly larger than any exponential function, but strictly smaller than
any double-exponential function. The technically subtle point consists in showing that no composition of
integer sequences in $\mathbb{S}_2$ can posess such a growth rate (work in preparation).

\begin{example}
\label{n_pui_n}
Let us consider the sequence $(D_n)_{n \in \NN}$ such that
$$D_0=0 \mbox{ and } \forall n\geq 1,\;\;D_n=n^n.$$
This sequence belongs to $\mathbb{S}_4$ and does not belong to $\FRAC(\mathbb{S}_3)$.
\end{example}
{\bf Proof for level 4}\\
The maps $f(n):= n$ and $g(n):=n$ have level $1$, hence $f \in \mathbb{S}_4$ and $g \in \mathbb{S}_3$.
By Proposition 66 of \cite{Fra-Sen06}, if $f \in \mathbb{S}_{k+1}$ and $g \in \mathbb{S}_k$ with $k \geq 3$, then
$n \mapsto f(n)^{g(n)} \in \mathbb{S}_{k+1}$. It follows that, for the above maps $f,g$ , $D=f^g \in \mathbb{S}_4$.  

\noindent{\bf Decomposition}\\
We take below the convention that $0^0=0$.\\
Let us exhibit two mappings $f \in \mathbb{S}_2(\NN,\{ a,b,c \}^*), g \in \mathbb{S}_3(\{ a,b,c \}^*,\NN)$
such that
\begin{equation}
\forall n \in \NN, D_n = g(f(n)).
\label{eq-decomp-nton}
\end{equation}
We define these mappings by two systems of recurrent relations.
Let $f,A,C: \NN \rightarrow \{ a,b,c \}^*$ such that
\begin{eqnarray}
\label{eq-system-for-f}
f(0) = b, & A(0) = a, & C(0) = c,\nonumber \\
f(n+1) = A(n)\cdot f(n) \cdot B(n),&  A(n+1) = A(n), & C(n+1) = C(n). 
\end{eqnarray}
Let us introduce a finite alphabet $X=\{x,y\}$ and define mappings $H,K,K',P: \{ a,b,c \}^* \rightarrow \HOM(X^*,X^*)$
by the recurrence relations:
\begin{eqnarray}
\label{eq-system-for-H}
H(\varepsilon)= \Id, & K(\varepsilon) = [x,xy], & K'(\varepsilon) = [x,\varepsilon],\;\;P(\varepsilon) = [y,x],\nonumber \\
H(au) = H(u) \circ H(u),& H(bu) = P(u) \circ H(u) \circ K'(u), & H(cu) = H(u) \circ K(u), \nonumber\\
K(au) = K(u),& K(bu) = K(u), & K(cu) = K(u), \nonumber\\
K'(au) = K'(u),& K'(bu) = K'(u), & K'(cu) = K'(u),\nonumber\\
P(au) = P(u),& P(bu) = P(u), & P(cu) = P(u).
\end{eqnarray}
where the bracketed notation $[w,w']$ designates the homorphism $X^* \to X^*$ that maps $x$ to $w$ and $y$ to $w'$. 
We finally define $g$ by:
$$\forall u \in \{ a,b,c \}^*,\;\; g(u) = H(u)(x).$$
One can check that, for every integers $p,q,n\in \NN$: 
$$H(c^q) = [x,x^qy],\;\;H(bc^q) = [x^q,x],\;\;H(a^pbc^q) = [x^{q^p} ,x]$$
$$f(n) = a^n b c^n.$$
$$g(a^n b c^n) = H(a^n b c^n)(x) = x^{n^n}$$
$$g(f(n)) = x^{n^n} \equiv n^n$$
The system (\ref{eq-system-for-f}) shows that $f$ has level 2. The system (\ref{eq-system-for-H}) is a system of recurrent relations in $\langle \HOM(\{a,b,c\}^*,\{a,b,c\}^*),\circ,{\rm Id}\rangle$, which implies, by point (2) of Theorem \ref{characterisation_level3}, that $g$ has level 3.\\
This decomposition also proves that $D$ has level $4$: by an adaptation of Proposition 70 of \cite{Fra-Sen06} to arbitrary alphabets),
the composition $D=f \circ g$ of a map of level $2$ by a map of level $3$, has level $4$.\\
{\bf Proof for not level 3:}\\
The proof of theorem 16 of \cite{Cad-Maz-Pap-Pil-Sen20} shows that, for every $k \geq 1$, the sequences $D_n,D_{n+1},\ldots,D_{n+k-1}$ are algebraically independent over $\bbbQ$. Theorem \ref{alg_dependency_level3} then shows that it cannot have level 3, in the strong sense that: $D \notin \FRAC(\mathbb{S}_3(\bbbn,\bbbn))$.

\section{Perspectives}
\label{perspectives7}
\subsection{Related works}
In \cite{Eng-Vog86}[theorem 8.12 p.366] it is proved that, for every $k \geq 1$
$$\D_t\CFT(\TR)^k = \D_t\RT(\Pil^k(\TR)).$$
The lefthand-side denotes the compositions of $k$ total deterministic macro-transductions (from trees to trees).\\
The righthand-side denotes the 
the total deterministic transductions(from trees to trees) of level $k+1$ (i.e. with indexes which are $k$-pushdowns of trees).\\
Thus, our theorem \ref{characterisation_level3}, is an analogue of this previous result of \cite{Eng-Vog86}, but for string-to-string transductions and for the level $k=2$.
The analogue theorem for string-to-string transductions of an arbitrary level $k \geq 1$ is stated by Theorem 6 of the extended abstract \cite{Sen07}. The full proof is in preparation and should appear soon.\\

Note that our theorem \ref{characterisation_level3} does not follow immediately from \cite{Eng-Vog86}[theorem 8.12 p.366] for two reasons:\\
R1- the two factors of the decomposition provided by this previous statement are
{\em string-to-tree} and {\em tree-to-string}, while we provide two factors that are {\em string-to-string} transductions.\\
R2- The fact that every total map $f \in \mathbb{S}_{3}(A^*,B^*)$ also belongs to $\D_t\RT(\Pil^2(\TR))$ (with domain and range included in unary trees) is true; but the only proof that we know
consists in applying $(1) \Rightarrow (4)$ of theorem \ref{characterisation_level3} and then \cite{Fra-Sen06}[Proposition 53 p. 384], while the proof of this proposition constructs a
level 3 pushdown-automaton {\em without any } $\push_3$ move (i.e. $\push$ move into the innermost stack); thus this pda has a memory-structure which is
$\Pd(\Gamma,\Pd(\Gamma,\Pd(\Gamma,\TR)))$ where $\Gamma$ is a finite alphabet 
and  $\TR$ is the structure of trees over a graded alphabet where the arities of the symbols are $0$ or $1$, see definition of \cite{Eng-Vog86}[p. 253, lines 39-40].\\

We think that the analogous result for {\em integer-to-integer} transductions is false, since the factorial map has level 3 (see example \ref{example:fact}),
as a string-to-string transduction, but we believe it 
to be non-decomposable into two {\em integer-to-integer} factors i.e. into two linear integer recurrent sequences (work in preparation).\\

The notion of Two-Way finite State Transduction with $k$ nested Pebbles ($\DPT_k$ for short) is defined in \cite{Eng-Man02}.
It turns out that $\D_t\PebT_k(A^*,B^*)$ (i.e. the total maps in the class $\DPT_k$ ) is included in $\mathbb{S}_{k+2}(A^*,B^*)$ (see definition \ref{def_Sk}).
This can be shown by the following arguments\footnote{Personnal communication of L{\^{e}} Th{\`{a}}nh Dung Nguy{\^{e}}n}:\\
-  by \cite{Eng-Man03}[theorem 35, p. 677], every transduction of $\D_t\PebT_k(A^*,B^*)$ belongs to $\D_t\CFT(\TR)^{k+1}$ (the compositions of $k+1$ macro-transductions)\\
-  by \cite{Eng-Vog86}[theorem 8.12 p.366], every transduction of  $\D_t\CFT(\TR)^{k+1}$ belongs to $\D_t\RT(\Pil^{k+1}(\TR))$\\
-  thus, every transduction of $\D_t\PebT_k(A^*,B^*)$ is computable by some tree-to-tree pushdown-transducer of level $k+2$ and has unary input;
one can easily see that such a transducer can be simulated by some word-to-word pushdown-transducer of level $k+2$.

For such maps a decomposition theorem similar to Theorem 6 of the extended abstract \cite{Sen07} has been  proved:
every map of $\D\PebT_k$ is a composition of $k$ maps of $\D\PebT_1$ (\cite{Eng-Man02}[theorem 1 p.237]).\\

Further progress on the hierarchy of transductions $(\DPT_k)_{k \in \bbbN}$ has been achieved in \cite{Boj22,Boj23} and \cite{Kie-Tit-Pra23}.
\subsection{Perspectives}
All the open problems mentionned at the end of \cite{Fra-Sen06} remain unsolved at the moment and deserve interest.\\
Let us mention some open problems which are specific to level 3.\\
1- The convolution operation is known to preserve $\mathbb{S}_2$ and it is shown in Proposition 67 of \cite{Fra-Sen06}
that, for every integer $k \geq 2$,  if $f \in \mathbb{S}_{k+1}, g \in \mathbb{S}_{k}$ then
$f * g \in \mathbb{S}_{k+1}$.\\
We wonder if the level $3$ is closed under convolution. The same question is open for all levels $k \geq 3$.\\
2- The equality problem for sequences in $\mathbb{S}_{3}(A^*, \NN)$ (for every finite alphabet $A$) is decidable (\cite{Sen07}).
Let us say that $(u(n))_{n \in \NN},(v(n))_{n \in \NN}$ are {\em almost}-equal (which we denote by $u=_a v$)  iff
$$ \{ n \in \NN \mid u_n \neq v_n\} \mbox{ is finite}.$$
The following decision problem is not known to be decidable nor undecidable:\\
{\bf Instance}: two sequences $u,v \in \mathbb{S}_3$\\
{\bf Question}: $u=_a v ?$\\
Note that the so-called ``Skolem problem'' (see below) is recursively reducible to this problem.
Skolem problem (usually credited to \cite{Sko34}) is the following decision problem:\\
{\bf Instance}: two sequences $u,v \in \mathbb{S}_2$\\
{\bf Question}: $\exists n \in \NN, u(n)= v(n) ?$\\
Skolem's problem is reducible (by a many-one recursive reduction) to the almost-equivalence problem for sequences of level 3,
by the following arguments.\\
Given  $u,v \in \mathbb{S}_2$, let us consider the sequence
$$w(n) := \prod_{i=0}^n (u(n) -v(n)).$$
This sequence $w$ belongs to ${\cal D}_3$ i.e. has the form
$$w(n)= w^+(n) -w^-(n)$$
where $w^+, w^- \in \mathbb{S}_3$ (see the sketch of proof below).\\
But
$$[\exists n \in \NN,\;u(n)=v(n)] \Leftrightarrow
[w^+ =_a w^-].
$$
\begin{sketch}
The sequence $w$ fulfills the recurrence relation
$$w(n+1) = w(n) \cdot (u(n+1) -v(n+1)).$$
which has the form
$$w(n+1) = P(w(n))$$
for some polynomial $P$ with coefficients in ${\cal D}_3$.
From Theorem 96 of \cite{Fra-Sen06}, step 1 of the proof, it follows that 
this sequence $w$ belongs to ${\cal D}_3$.
\end{sketch}
3- A natural decision problem is the following:\\
{\bf Instance}: a sequences $u \in \mathbb{S}_{k+1}$\\
{\bf Question}: $u \in \mathbb{S}_{k}?$\\
which we can name the {\em level-minimization problem}.
A positive solution is given in \cite{Doueneau23} for several sub-hierarchies of $(\D\PebT_k)_{k \in \bbbN}$.
It is not known whether this problem is decidable or not  for $(\D\PebT_k)_{k \in \bbbN}$ nor for $(\mathbb{S}_{k})_{k \in \bbbN}$.\\
Example \ref{n_pui_n} suggests that there is no simple property that characterizes  level 3 inside $\mathbb{S}_{4}$; for example, a low growth-rate does not imply
a low level in the hierarchy $(\mathbb{S}_{k})_{k \in \bbbN}$.
\paragraph{Aknowledgments}
I thank L{\^{e}} Th{\`{a}}nh Dung Nguy{\^{e}}n and Vincent Penelle for fruitful discussions
and useful information on these subjects.

\bibliographystyle{alpha}
\bibliography{seqlevel3}
\end{document}